\documentclass[a4paper,10pt,leqno]{amsart}
\usepackage{amsmath}
\usepackage{amsthm}
\usepackage{amssymb}
\usepackage{stmaryrd}
\usepackage{url}
\usepackage{here}
\usepackage[all]{xy}
\usepackage{comment}
\usepackage{mathtools}
\usepackage{ascmac}
\usepackage{ulem}
\usepackage{framed}
\usepackage{graphicx}
\usepackage{nicematrix}
\usepackage{lmodern}
\usepackage{caption}
\usepackage{subcaption}
\usepackage{footnote}
\usepackage{amscd}
\usepackage{mathdots}
\usepackage{titlesec}
\usepackage{titletoc}
\mathtoolsset{showonlyrefs}

\usepackage{xcolor}
\definecolor{webgreen}{rgb}{0,.5,0}
\definecolor{webbrown}{rgb}{.6,0,0}
\definecolor{RoyalBlue}{cmyk}{1, 0.50, 0, 0}
\usepackage[colorlinks=true, breaklinks=true, urlcolor=webbrown, linkcolor=RoyalBlue, citecolor=webgreen, backref=page]{hyperref}

\titlecontents{section} % <section type>
[0pt]                 % <left indent>
{\bfseries}            % <above code>
{\thecontentslabel\quad} % <numbered entry format>
{}                    % <numberless entry format>
{\titlerule*[0.5em]{.}\contentspage} % <filler page format>

\titlecontents{subsection} % <section type>
[5pt]                 % <left indent>
{}            % <above code>
{\thecontentslabel\quad} % <numbered entry format>
{}                    % <numberless entry format>
{\titlerule*[0.5em]{.}\contentspage} % <filler page format>

\titlecontents{subsubsection} % <section type>
[10pt]                 % <left indent>
{}            % <above code>
{\thecontentslabel\quad} % <numbered entry format>
{}                    % <numberless entry format>
{\titlerule*[0.5em]{.}\contentspage} % <filler page format>

\titlecontents{paragraph} % <section type>
[15pt]                 % <left indent>
{}            % <above code>
{\thecontentslabel\quad} % <numbered entry format>
{}                    % <numberless entry format>
{\titlerule*[0.5em]{.}\contentspage} % <filler page format>

\titlespacing*{\section}{0pt}{5.5ex plus 1ex minus .2ex}{4.3ex plus .2ex}
\numberwithin{equation}{subsection}

\titleformat{\section}{\centering\normalfont\large\scshape}{\thesection}{1em}{}

\titleformat{\subsection}{\normalfont\normalsize\bfseries}{\thesubsection}{1em}{}

\titleformat{\subsubsection}{\normalfont\normalsize\bfseries}{\thesubsubsection}{1em}{}

\titleformat{\paragraph}{\normalfont\normalsize\it}{\theparagraph}{1em}{}

\theoremstyle{plain}
\newtheorem{thm}{Theorem}[section]
\newtheorem{lem}[thm]{Lemma}

\newtheorem{RHP}{Riemann-Hilbert Problem}

\theoremstyle{definition}

\theoremstyle{remark}
\newtheorem*{rem}{Remark}

\newtheoremstyle{case}{}{}{}{}{}{:}{ }{}
\theoremstyle{case}

\newcommand{\C}{\mathbb C}   
\newcommand{\R}{\mathbb R}

\newcommand{\Z}{\mathbb Z}

\DeclareMathOperator{\res}{res}
\DeclareMathOperator{\sign}{sgn}

\DeclareMathOperator{\arccosh}{arccosh}

\newcommand{\bp}{\begin{pmatrix}} 
\newcommand{\ep}{\end{pmatrix}} 
\renewcommand{\O}{\mathcal{O}}
\renewcommand{\Re}{\operatorname{Re}}
\renewcommand{\Im}{\operatorname{Im}}

\begin{document}

\title[Singular asymptotics of the sinh-Gordon equation]{The non-linear steepest descent approach to the singular asymptotics of the sinh-Gordon reduction of the Painlev\'e III equation}

\author[Its]{Alexander R. Its}
\address{(A.R. Its) Department of Mathematical Sciences, 
Indiana University Indianapolis,
402 N. Blackford St., Indianapolis, IN 46202 USA}
\email{aits@iu.edu}

\author[Miyahara]{Kenta Miyahara}
\address{(K. Miyahara) Department of Mathematical Sciences, 
Indiana University Indianapolis,
402 N. Blackford St., Indianapolis, IN 46202 USA}
\email{kemiya@iu.edu}

\author[Yattselev]{Maxim L. Yattselev}
\address{(M.L. Yattselev) Department of Mathematical Sciences, 
Indiana University Indianapolis,
402 N. Blackford St., Indianapolis, IN 46202 USA}
\email{maxyatts@iu.edu}

\begin{abstract}
Motivated by the simplest case of tt*-Toda equations, we study the large and small $x$ asymptotics for \( x>0 \) of real solutions of the sinh-Godron Painlev\'e III($D_6$) equation. These solutions are parametrized through the monodromy data of the corresponding Riemann-Hilbert problem. This unified approach provides connection formulae between the behavior at the origin and infinity of the considered solutions. 
\end{abstract}

\date{\today}

\let\ds\displaystyle

\maketitle 

\setcounter{tocdepth}{3}
\tableofcontents

%%%%%%%%%%%%%%%%%%%%%%%%%%%%%%%%%%%%%%%%%%%%%%%%%%%%%%%%%%%%%%%%%%%%%%%%%%%%%
\setlength{\parskip}{6pt}

\section{Introduction}

The tt* (topological—anti-topological fusion) equations were introduced by Cecotti and Vafa to describe certain deformations of super-symmetric quantum field theories \cite{CV1, CV2, CV3}. Our interest lies in the special case known as tt*-Toda equations subject to the anti-symmetry and radial symmetry conditions \cite{GuLi14, GIL1, GIL2, GIL3}. In the simplest case (tt*-Toda equation of type $A_1$) it becomes
\begin{align}
(w_0)_{rr} + \frac{(w_0)_r}{r} = 4 \sinh (4w_0), \label{tt*-Toda n=1}
\end{align}
where $w_0 : \R_{>0} \to \R$, $r \mapsto w_0(r)$, and \( \R_{>0} :=(0,\infty) \). It is well known that the above equation is a special reduction of the Painlev\'e III($D_6$) equation, given, in the canonical form, by 
\begin{align}
    w_{xx} = \frac{w_x^2}{w} - \frac{w_x}{x} + \frac{\alpha w^2 + \beta}{x} + \gamma w^3 + \frac{\delta}{w}, \label{PIIID6}
\end{align}
where $w : \C \to \C$, $x \mapsto w(x)$ and $\alpha, \beta, \gamma, \delta$ are some complex parameters such that $\gamma \delta \neq 0$. If we choose $\alpha = \beta = 0$ and $\gamma = -\delta = 1/4$, then any of the four functions $w(x) =\pm e^{\pm u(x)/2}$ solves \eqref{PIIID6} if \( u(x) \) is a solution of the sinh-Gordon Painlev\'e~III equation
\begin{align}
    u_{xx} + \frac{u_x}{x} = \sinh u, \label{sinh-gordon PIII}
\end{align}
which is identical to \eqref{tt*-Toda n=1} subject to a minor transformation \( 4w_0(r) = u(4r)\). A further transformation $u(x) = i \pi - i v(x)$ brings \eqref{sinh-gordon PIII} into the radial symmetric reduction of the elliptic sine-Gordon equation:
\begin{align}
    v_{xx} + \frac{v_x}{x} + \sin v = 0, \label{sine-gordon PIII}
\end{align}
which is usually called the sine-Gordon Painlev\'e III equation and plays an important role in the study of the 2D Ising model \cite{BMW, MTW}. 

Due to our interest in tt*-Toda equations, we study asymptotics (as \( x\to+\infty \) and \( x\to 0^+\)) of {\it real} solutions of \eqref{sinh-gordon PIII}. More precisely, recall that solutions of \eqref{PIIID6} can have non-polar singularities only at $0$ and $\infty$ while polar singularities in $\C^* := \C \setminus \{0\}$ do depend on the considered solution (this is the so-called Painlev\'e property). Hence, solutions of the sinh-Gordon reduction \eqref{sinh-gordon PIII} of \eqref{PIIID6} can have movable logarithmic singularities. It also can be readily verified that solutions of \eqref{PIIID6} can only have simple poles. Thus, going around any logarithmic singularity (via the complex plane) of a solution of \eqref{sinh-gordon PIII} will result in the addition of \( 2 \pi i \) to this solution. So, when we speak of real solutions, we ignore all the acquired additions of \( 2\pi i\) (since \eqref{sinh-gordon PIII} does not change when \( u(x) \) is replaced by \( u(x) + 2\pi i n \), \( n\in \Z \), thus obtained functions are still local solutions of \eqref{sinh-gordon PIII}). Clearly, real solutions of \eqref{sinh-gordon PIII} correspond to the real (on \(\R_{>0}\)) solutions of \eqref{PIIID6}.

Investigations of this type are not new. A lot of information on the history of the sine/sinh-Gordon Painlev\'e III equation can be found in the monograph \cite{FIKN}. In particular, McCoy, Tracy, and Wu \cite{MTW} constructed a one-parameter family of solutions of \eqref{PIIID6} with \( \alpha (-\delta)^{1/2} + \beta \gamma^{1/2} = 0 \) that remain bounded for large \( x \) and described their behavior as \( x\to+\infty \) and \( x\to0^+\). These solutions are in fact smooth on the whole \( \R_{>0}\) and are called \textit{global solutions}, see \cite{GIL2}. Later, Novoksh\"enov \cite{Novok4} asymptotically identified the locations of singularities of non-smooth real solutions of \eqref{sinh-gordon PIII} parametrized by their behavior at the origin: \( u(x)= r\ln x+s+ \O(x^{2-|r|})\) for some \( s,|r|<2 \).

A substantial amount of work has been dedicated to describing asymptotics of the solutions of \eqref{sine-gordon PIII} similarly parametrized by their behavior at the origin via
\begin{equation}
\label{r s pairs}
v(x)= r \ln x+ s + \O\big(x^{2-|\Im(r)|}\big), \quad x\to 0^+.
\end{equation}
In \cite{Novok2, Novok3}, Novoksh\"enov studied asymptotics as \( x\to+\infty \) of the real solutions, which correspond to purely imaginary solutions of \eqref{sinh-gordon PIII}. In \cite{Novok07boutrouxforSGP3, Novok08boutrouxforSGP3}, assuming \( r,s\in\C\), \( |\Im(r)|<2 \), he described asymptotics of \eqref{r s pairs} for large \( |x| \), \( 0<\arg(x)<\frac\pi2\), in terms of the Jacobi $\mathrm{sn}$ function (when \( r,s\) are purely imaginary, this provides asymptotics of the real solutions of \eqref{sinh-gordon PIII} in the first quadrant). Asymptotics as \( x\to+\infty \) of \eqref{r s pairs} for a range of parameters \( r,s \) including purely imaginary ones with \( |\Im(r)|<2 \), can be found in the introduction of \cite{ItsProk}. Even though these asymptotics can be obtained from the work of Kitaev \cite{Kitaev2} outlined in the next paragraph, they were re-derived by the authors using the Riemann-Hilbert approach employed in this paper. Let us point out that the small $x$ asymptotics \eqref{r s pairs} does not describe all the behavior of solutions near $0^+$, see Theorem~\ref{Niles result} further below, which is based on the work of Niles \cite{Niles}. 

In \cite{Kitaev2}, Kitaev used the isomonodromy theory and the WKB method to find the connection formulae for the small and large parameter asymptotics of the solutions of the degenerate Painlev\'e V equation, which bijectively corresponds to Painlev\'e III($D_6$). These results then in principle allow one to extract the large and small $x$ asymptotics of the solutions of \eqref{sinh-gordon PIII} via certain non-trivial changes of variables. However, these transformations are rather complicated. More importantly, the approach taken in \cite{Kitaev2} does not readily generalize to tackling the asymptotics of tt*-Toda equations of type \( A_n \) for general \( n>1 \), which is our overarching goal. On the other hand, the approach we take here of parametrizing solutions via the monodromy data of the associated linear system and then applying Deift-Zhou steepest descent techniques \cite{DZ} as \( x\to+\infty \) and techniques of Niles \cite{Niles} as \( x\to0^+ \) to the corresponding Riemann-Hilbert problem seem to be much more promising in this regard, see \cite{GIKMO}.

%%%%%%%%%%%%%%%%%%%%%%%%

\section{Main Results}

As standard in the theory of Painlev\'e transcendents, we parameterize the solutions by the monodromy data of the corresponding Riemann-Hilbert problem. Section~\ref{section:2} contains a detailed description of how this parameterization is achieved. At this point, it suffices to say that the monodromy data can be represented by a single parameter \( p \) that belongs to \( \overline\C\setminus\overline{\mathbb D} \), where \( \mathbb D \) is the open unit disk centered at the origin. Our main result is the following theorem.

\begin{thm}
\label{main thm}
To each finite \( p \), \( |p|>1 \), there corresponds one real solution of the sinh-Gordon Painlev\'e III equation \eqref{sinh-gordon PIII} on \( \R_{>0} \) and
        \begin{align}
            u(x) = 2 \, \ln \left(\frac{1 - \sin \left( x - \kappa\ln x + \phi \right)+\O(x^{-1})}{\left|\cos \left( x - \kappa \ln x + \phi \right) +\O(x^{-1}) \right|} \right), \label{asymp for B neq 0}
        \end{align}
as $x \to +\infty$, where the error term is uniform in \( x \),
\begin{equation}
\label{kappa phi}
\kappa:= \frac{1}{2\pi} \ln\big(|p|^2 -1\big), \quad \text{and} \quad \phi := -\kappa \ln 4 + \arg \left( p\,\Gamma\left(\frac12+i\kappa\right) \right).
\end{equation}
To \( p=\infty \) there corresponds a one-parameter family of real solutions of the sinh-Gordon Painlev\'e III equation \eqref{sinh-gordon PIII} on \( \R_{>0} \) such that 
\begin{align}
    u(x) = \frac{2s^\R}\pi K_0(x) + \O\left(\frac{e^{-2x}}{\sqrt{x}} \right),
    \label{exp decay}
\end{align}
as \( x\to+\infty \), where \( s^\R\in\R \) is the parametrizing (Stokes) parameter and \( K_0(x) \) is the modified Bessel function of the second kind.

\end{thm}

\begin{rem}
Recall, see \cite[Equation (10.40.2)]{NIST:DLMF}, that for large \( x \) it holds that
\[
K_0(x) \sim \frac{e^{-x}}{\sqrt {2\pi}} \sum_{n=0}^\infty \frac{(-1)^n}{2^n} \frac{\Gamma^2(n+1/2)}{\Gamma(n+1)} \frac{1}{x^{n+1/2}}.
\]    
\end{rem}

\begin{rem}
Formula \eqref{asymp for B neq 0} implies that if we label the singularities of \( u(x)\) by \( x_n \) in the increasing order, then for all natural numbers \( n \) large enough it holds that
\begin{equation}
\label{asymptotic of x_n}
    x_n = \pi\left( n-\frac12 \right) + \kappa \ln(4\pi n) - \arg \left( p\,\Gamma\left(\frac12+i\kappa \right) \right) + \O\left( \frac{\ln n}n \right).
\end{equation}
\end{rem}

\begin{rem}
Denote by \( \Xi_\epsilon \) the complement of \( \cup_n (x_n-\epsilon,x_n+\epsilon) \) in \( \R_{>0} \). Then,
\begin{align*}
u(x) & = 2 \, \ln \left(\frac{1 - \sin \left( x - \kappa\ln x + \phi \right)}{\left|\cos \left( x - \kappa \ln x + \phi \right)\right|} \right) + \O_\epsilon(x^{-1})    \\
& = \ln \left(\frac{1 - \sin \left( x - \kappa\ln x + \phi \right)}{1 + \sin \left( x - \kappa\ln x + \phi \right)} \right) + \O_\epsilon(x^{-1}),
\end{align*}
where the error term is uniform on \( \Xi_\epsilon \) for all \( x \) large (i.e., locally uniformly in \( \Xi_0 \)).
\end{rem}

\begin{rem}
    As mentioned in the introduction, the location of poles of \( w \), the corresponding solution of   \eqref{PIIID6}, has been investigated in \cite{Novok4} when \(|\Im(p)|<1\). Novoksh\"enov considered \eqref{sinh-gordon PIII} with \( \sinh u \) replaced with \( 4 \sinh u \), in which case the solution is given by \( u(2x) \). Moreover, he was only interested in the singularities if \( u \) that lead to the poles of \( w \) (and not zeros), i.e., in \( x_{2n} \). These points were denote by \( a_n \) in \cite[Equation (0.9)]{Novok4} (this formula has a typo, \( \pi(n-\frac12) \) needs to be replaced by \( \pi (n+\frac14) \) as evident from \cite[Equations (1.1) and (4.9)]{Novok4}). Then one can readily verify that \( a_n \) and \( \frac12 x_{2n} \) have exactly the same asymptotic behavior upon noticing that \( r,s \) in \cite[Equation (0.2)]{Novok4} are equal to \( \gamma,\rho_p+\gamma\ln2 \), see Theorem~\ref{Niles result} further below, and \( A,B \) from \cite[Equation (0.8)]{Novok4} can be expressed as
    \[
    A = \frac\pi{\sqrt{1-\Im(p)^2}} \frac{\sqrt{|p^2|-1}}{\Re(p)+\sqrt{1-\Im(p)^2}} \quad \text{and} \quad AB = \frac{\pi^2}{1-\Im(p)^2}.
    \]
\end{rem}

Due to \cite{Niles}, one can describe the asymptotic behavior near $0$ of the real solutions of \eqref{sinh-gordon PIII} using the above parameter $p$, which provides the connection formulae between asymptotics near infinity and the origin for each real solution.

\begin{thm}\label{Niles result}
When \( p=\infty\), the corresponding real solutions of the sinh-Gordon Painlev\'e III equation \eqref{sinh-gordon PIII} parametrized by the Stokes parameter \( s^\R \) as in Theorem~\ref{main thm} admit the following asymptotic behavior on \( \R_{>0} \) near the origin:
\begin{equation}
\label{asymptotic near zero}
 u(x) = \begin{cases}
     \displaystyle \gamma \ln x + \rho_\infty + \O\big(x^{2 - |\gamma|}\big), & |s^\R|<2, \bigskip \\
     \displaystyle 2\sigma \ln \left[-\frac{x}{2} \left(\ln \frac{x}{8} + \delta_\infty \right)\right] + \O\left(\frac x{\ln x}\right), & |s^\R| = 2, \bigskip \\
     \displaystyle 2 \sigma \ln \left| \frac{x}{4t} \sin \left( 2t \ln \frac{x}{8} + \Delta_\infty \right) \right| + \O(x^2), & |s^\R|>2,
 \end{cases}
 \end{equation}
 where the error terms hold as \( x\to 0^+ \) and, in the third case, it is locally uniform in the complement of \( \cup_{n\geq 0} 8e^{-\frac{\pi n+\Delta_\infty}{2t}}\) (asymptotic locations of the singularities),
 \[
 \gamma := - \frac{4}{\pi} \arcsin \left( \frac{s^{\R}}{2} \right), \quad t := \frac{1}{\pi} \arccosh \left( \frac{|s^{\R}|}{2} \right), \quad \sigma := - \mathrm{sgn}(s^{\R}),
 \]
 \( \delta_\infty \) is the Euler's constant, and
 \[
 \rho_\infty := -3 \gamma \ln 2 + \ln\left[ \frac{ \Gamma^2 \left(\frac{1}{2} - \frac{\gamma}{4} \right)}{\Gamma^{2} \left(\frac{1}{2} + \frac{\gamma}{4} \right)}\right], \quad \Delta_\infty := - 2 \arg \Gamma(1 + it).
 \]
For finite $p$, \( |p|>1 \), the corresponding real solution of the sinh-Gordon Painlev\'e III equation \eqref{sinh-gordon PIII} admits asymptotic behavior on \( \R_{>0} \) near the origin as in \eqref{asymptotic near zero}, this time depending on the value of \( s^\R := -2\Im(p) \) and with constants \( \rho_\infty,\delta_\infty,\Delta_\infty\) replaced by
\[
\rho_p := \rho_\infty + \ln \left[ \frac{|p|^2-1}{q_p^2} \right], \quad \delta_p := \delta_\infty + \frac{\sigma\pi}{2q_p}, \quad \Delta_p := \Delta_\infty + \sigma\arg(q_p),
\]
where \( \displaystyle q_p := \Re(p) + \sqrt{1-\Im(p)^2} \in\{z:\Im(z)\geq0\} \).
\end{thm}

\begin{rem}
As shown by \cite{GIL2,GIL4}, the case $p = \infty$ and $s^{\R} \in [-2, 2]$ corresponds to the globally smooth solutions on \( \R_{>0} \). Theorems~~\ref{main thm} and~\ref{Niles result} show that these solutions admit exponentially decaying behavior \eqref{exp decay} near $x = \infty$ and logarithmic behavior near $x = 0$ as in the first two lines of \eqref{asymptotic near zero}.
\end{rem}

\begin{rem}
Theorems~~\ref{main thm} and~\ref{Niles result} provide a connection between the behavior at infinity and around the origin along \( \R_{>0}\) of the corresponding solutions Painlev\'e III equation \eqref{PIIID6}, but only up to a sign since \( w(x) = \pm e^{\pm u(x)/2}\) and \( u(x)/2 \) is defined up to addition of integer multiples of \( \pi i \). However, the proofs of these theorems allow us to determine the sign uniquely, which we do in Appendix~\ref{connection P3}.
\end{rem}

\begin{rem}
Recent paper \cite[Section 4.1]{YuqiLi} conjectures, based on numerical analysis, singular asymptotics near the origin of exponentially decaying near infinity solutions \( w_0(r),w_1(r)\) of the tt*-Toda equation of type $A_3$:
\begin{align}
\begin{dcases}
    (w_0)_{rr} + \frac1r (w_0)_r = -2 e^{2(w_1 - w_0)} + 2 e^{4 w_0}, \\
    (w_1)_{rr} + \frac1r (w_1)_r = -2 e^{-4 w_1} + 2 e^{2(w_1 - w_0)},
\end{dcases} \label{tt*-toda A3}
\end{align}
where $w_0,w_1$ are real-valued functions of \( r \) on $\R_{>0}$. Clearly, a further assumption $w_0 + w_1 = 0$ reduces \eqref{tt*-toda A3} to \eqref{tt*-Toda n=1}. Theorem \ref{Niles result} proves the conjecture in this degenerate case, see Appendix~\ref{conj by Li} for details.
\end{rem}

As standard in the theory of Painlev\'e transcendents, we view \eqref{sinh-gordon PIII} as a compatibility condition of an overdetermined system of linear differential equations, a Lax pair. A Lax pair for the sinh-Gordon equation was first found by Flaschka-Newell \cite{FN}, see also \cite[Chpater 15]{FIKN} (the Lax pairs in these references lead to the one used below after minor changes of variables). Let a matrix function $\Psi: \C^* \times \C \to \mathrm{Mat}_{2}(\C)$, $(\lambda, x) \mapsto \Psi(\lambda, x)$, be such that
\begin{align}
\begin{dcases}
    \frac{d \Psi}{d \lambda} = \left( \frac{i x^2}{16} \sigma_3 + \frac{x u_x}{4 \lambda} \sigma_1 + \frac{i}{\lambda^2}Q \right) \Psi,\\
    \frac{d \Psi}{d x} = \left( \frac{u_x}{2} \sigma_1 + \frac{i x \lambda}{8} \sigma_3 \right) \Psi, 
\end{dcases}
    \label{Lax pair sinh-Gordob PIII}
\end{align}
where
\begin{equation}
    \label{defn of Q}
    Q = \begin{pmatrix}
        \cosh u & - \sinh u\\
        \sinh u & - \cosh u
    \end{pmatrix}.
\end{equation}
Then one can readily verify that the compatibility condition $\Psi_{\lambda x} = \Psi_{x \lambda}$ is equivalent to the sinh-Gordon Painlev\'e III equation.

\begin{rem}
As pointed out in the introduction, when \( x\in\R_{>0}\), equation \eqref{sinh-gordon PIII} could be seen as the tt*-Toda equation of type $A_1$, which admits a Lax pair representation originally discovered by Mikhailov in \cite{Mik}, see also \cite[Equations (1.3)-(1.4)]{GIL2},
\begin{align}
    \begin{dcases}
    \Tilde{\Psi}_{\zeta} = \left(- \frac{1}{\zeta^2} W -\frac{r(w_0)_{r}}{\zeta}\sigma_3  + r^2 W^{T} \right) \Tilde{\Psi}, \\
    \Tilde{\Psi}_{r} = \left( -(w_0)_{r}\,\sigma_3 + 2 r \zeta W^T \right) \Tilde{\Psi},
\end{dcases} \label{Lax pair for tt*}
\end{align}
where $\Tilde{\Psi}$ maps $(\zeta, r) \in \C^* \times \R_{>0}$ to $\Tilde{\Psi}(\zeta, r) \in \mathrm{Mat}_2 (\C)$ and  
\begin{align}
    W = \begin{pmatrix}
        0 & e^{-2w_0}\\
        e^{2w_0} & 0
    \end{pmatrix}.
\end{align}
Under a change of variables 
\begin{align}
    w_0(r) = \frac{1}{4} u(x), \quad r = \frac{1}{4} x, \quad \zeta = i \lambda,
\end{align}
the gauge transformation
\begin{align}
    \Psi(\lambda, x) = \frac{1}{\sqrt{2}}
    \begin{pmatrix}
        e^{-w_0} & e^{w_0}\\
        -e^{-w_0} & e^{w_0}\\
    \end{pmatrix} \tilde{\Psi}(\zeta, \Tilde{x})
\end{align}
defines a bijection between solutions of the Lax pair \eqref{Lax pair for tt*} for the tt*-Toda equation of type $A_1$ and solutions of \eqref{Lax pair sinh-Gordob PIII}, the Lax pair for the sinh-Gordon equation.
\end{rem}

\begin{rem}
Another important Lax pair for us is the Lax pair for the sine-Gordon equation \eqref{sine-gordon PIII} that can be found in \cite[Chapter 13]{FIKN} and \cite{Niles}:
\begin{align}
\begin{dcases}
    \frac{d \Phi}{d \lambda} = \left(-\frac{i x^2}{16} \sigma_3 - \frac{i x v_x}{4 \lambda} \sigma_1 + \frac{i}{\lambda^2} \begin{pmatrix}
        \cos v & i \sin v\\
        -i \sin v & - \cos v
    \end{pmatrix} \right) \Phi,\\
    \frac{d \Phi}{d x} = \left( - \frac{i x \lambda}{8} \sigma_3 - \frac{i v_x}{2} \sigma_1 \right) \Phi .\label{Niles pair}
\end{dcases}
\end{align}
One can readily check that the change of variables $u(x) = i \pi - i v(x)$ and gauge transformations $\Psi (\lambda) = i \sigma_1 \Phi (\lambda)$ or $\Psi (\lambda) = \sigma_1 \Phi (\lambda) \sigma_1$ provide correspondence between \eqref{Niles pair} and \eqref{Lax pair sinh-Gordob PIII}.    
\end{rem}

%%%%%%%%%%%%%%%%%%%%%%%%%%%%%%%%%%

Acknowledgments: The first author was partially supported by NSF grant DMS: 1955265, by RSF grant No. 22-11-00070, and by a Visiting Wolfson Research Fellowship from the Royal Society. The second author was partially supported by a scholarship from the Japan Student Services Organization and the Hokushin Scholarship Foundation. The third author was partially supported by a grant from the Simons Foundation, CGM-706591.

\section{Proof of Theorem~\ref{Niles result}}
\label{section:2}

In this section, we show how to formulate a Riemann-Hilbert problem for a solution of \eqref{Lax pair sinh-Gordob PIII}. Then we state an analogous Riemann-Hilbert problem for a solution of \eqref{Niles pair} and show how to connect the monodromy data between these two problems. Theorem~\ref{Niles result} is then automatically obtained from \cite[Theorems~7--9]{Niles}.

\subsection{Direct Monodromy Problem for sinh-Gordon Equation}

Consider the first equation of the Lax pair \eqref{Lax pair sinh-Gordob PIII}, that is,
\begin{align}
    \frac{d \Psi}{d \lambda} = A(\lambda, x) \Psi, \quad A(\lambda,x) := \frac{i x^2}{16} \sigma_3 + \frac{x u_x}{4 \lambda} \sigma_1 + \frac{i}{\lambda^2}Q, \label{first eq}
\end{align}
where \( Q \) is given in \eqref{defn of Q}. For future use observe that the coefficient matrix $A(\lambda, x)$ satisfies the following relations:
\begin{itemize}
    \item anti-symmetry: $\sigma_3 A(\lambda,x) \sigma_3 = A^T(- \lambda,x)$;
    \vspace{.1cm}
    \item inversion symmetry: $P_0 \sigma_3 A\big(-\frac{16}{\lambda x^2},x\big) \sigma_3 P_0^{-1} = \frac{\lambda^2 x^2}{16} A(\lambda,x)$;
    \vspace{.1cm}
    \item cyclic symmetry: $\sigma_1 A(\lambda,x) \sigma_1 = - A(- \lambda,x)$;
    \vspace{.1cm}
    \item reality: $\overline{A(\overline{\lambda},x)} = \sigma_1 A(\lambda,x) \sigma_1$.
\end{itemize}

We note that equation \eqref{first eq} has two singular points $\lambda = 0, \infty$, which are both irregular, and that the coefficients of $A(\lambda, x)$ near the leading order singularities are diagonalizable. Indeed, near infinity (after changing to a local variable \( \xi \), where \( \lambda = 1/\xi \)) it is \( -(ix^2/16)\sigma_3 \), which is already diagonal, and near zero it is
\begin{align}
    iQ = P_0 \, i\sigma_3 \, P_0^{-1}, \quad 
    P_0 = \begin{pmatrix}
        \cosh \frac{u}{2} & \sinh \frac{u}{2}\\
        \sinh \frac{u}{2} & \cosh \frac{u}{2}
    \end{pmatrix}.
    \label{defn of P0}
\end{align}
Thus, see \cite[Prop. 1.1]{FIKN}, {\it formal fundamental solutions} of \eqref{first eq} at $\lambda = 0, \infty$ must be of the form
\begin{align}
\begin{aligned}
    &\Psi_{f}^{(0)} (\lambda) = P_0 (I + \O(\lambda)) e^{-\frac{i}{\lambda} \sigma_3},\\
    &\Psi_{f}^{(\infty)} (\lambda) = (I + \O(1/\lambda)) e^{\frac{i x^2 \lambda}{16} \sigma_3}.
\end{aligned} \label{formal solution}
\end{align}

Formal fundamental solutions \eqref{formal solution} determine {\it canonical solutions} (genuine solutions that possess asymptotic behavior \eqref{formal solution} in the corresponding Stokes sectors) near $\lambda = 0, \infty$. This principle is explained in detail in Chapter 2 of \cite{FIKN}. We give only a brief list of facts here.
\begin{itemize}
\item Stokes sectors at $\zeta = 0$:
\begin{align}
\begin{aligned}
    &\Omega_{2k+1}^{(0)}(\lambda) = \left\{ \lambda \in \C^* \,|\, -(2k+1)\pi < \arg \lambda < -(2k-1)\pi \right\}, \\
    &\Omega_{2k}^{(0)}(\lambda) = \left\{ \lambda \in \C^* \,|\, -2k\pi < \arg \lambda < -2(k-1)\pi \right\}.
\end{aligned}
\end{align}
\item Stokes sectors at $\zeta = \infty$:
\begin{align}
\begin{aligned}
    &\Omega_{2k+1}^{(\infty)}(\lambda) = \left\{ \lambda \in \C^* \,|\, (2k-1)\pi < \arg \lambda < (2k+1)\pi \right\}, \\
    &\Omega_{2k}^{(\infty)}(\lambda) = \left\{ \lambda \in \C^* \,|\, 2(k-1)\pi < \arg \lambda < 2k \pi \right\}.
\end{aligned}
\end{align}
\end{itemize}
It is important for us to keep track of the argument of \( \lambda \) as we think of Stokes sectors as subsets of the universal covering of \( \C^* \). Notice that $\Omega_n^{(0)}(\lambda) = \Omega_n^{(\infty)}(\lambda^{-1})$ for $n \in \Z$. In each Stokes sector, there is a unique (canonical) solution  that satisfies
\begin{align}
\begin{aligned}
    &\Psi_n^{(0)}(\lambda) \sim \Psi_{f}^{(0)} (\lambda), \quad \lambda \to 0, \quad \lambda \in \Omega_n^{(0)}, \\
    &\Psi_n^{(\infty)}(\lambda) \sim \Psi_{f}^{(\infty)} (\lambda), \quad \lambda \to \infty, \quad \lambda \in \Omega_n^{(\infty)}.
\end{aligned}
\end{align}

Every solution of \eqref{first eq} and particularly all of the canonical solutions are analytically continued on the whole universal covering of $\C^*$. Hence, it holds that
\begin{align}
\label{riemann surface}
    \Psi_{n}^{(0)}(\lambda) = \Psi_{n-2}^{(0)}(\lambda e^{2\pi i}), \quad
    \Psi_{n-2}^{(\infty)}(\lambda) = \Psi_{n}^{(\infty)}(\lambda e^{ 2\pi i}),
\end{align}
where \( \Psi(\lambda e^{2\pi i}) \) is the short-hand notation for the analytic continuation of \( \Psi(\lambda) \) to the same point \( \lambda \) along the path that winds once around the origin in the counter-clockwise direction. 

Because any two solutions are connected via multiplication by a constant matrix on the right, there exists a connection matrix $E$ such that
\begin{align}
\label{connection matrix}
    \Psi_1^{(\infty)}(\lambda) = \Psi_1^{(0)}(\lambda) E.
\end{align}
Using cyclic symmetry of \( A(\lambda,x)\) one can readily verify that \( \sigma_1 \Psi(-\lambda) \sigma_1\) is a canonical solution if \( \Psi(\lambda)\) is a canonical solution. In particular, it holds that
\begin{align}
    \Psi_1^{(0)}(\lambda) = \sigma_1 \Psi_2^{(0)}(\lambda e^{-i\pi}) \sigma_1 \quad \text{and} \quad \Psi_1^{(\infty)}(\lambda) = \sigma_1 \Psi_2^{(\infty)}(\lambda e^{i\pi}) \sigma_1.
    \label{symmetry in terms of canonical soln}
\end{align}
Combining the above observations with \eqref{connection matrix}, we get that
\begin{align}
\label{connection matrix 2}
    \Psi_2^{(\infty)}(\lambda) = \Psi_2^{(0)}(\lambda e^{-2\pi i}) \sigma_1 E \sigma_1.
\end{align}

We define Stokes matrices as the ones connecting adjacent canonical solutions:
\begin{align}
    S_n^{(0, \infty)} = \left[ \Psi_n^{(0, \infty)}(\lambda) \right]^{-1} \Psi_{n+1}^{(0, \infty)}(\lambda), \label{defn of stokes mat}
\end{align}
where \( \lambda \) belongs to the upper half-plane if \( n \) is odd and to the lower half-plane if \( n \) is even (again, these matrices are necessarily constant). Due to the analytic continuation property, it holds that
\[
S_{2k+1}^{(0, \infty)} = S_1^{(0, \infty)} \quad \text{and} \quad S_{2k}^{(0, \infty)} = S_2^{(0, \infty)}.
\]
Taking $\lambda \to 0, \infty$ in \eqref{defn of stokes mat} for $\lambda$ in the appropriate Stokes sector and using the series expression of each canonical solution \eqref{formal solution} at those singularities, one can readily establish triangularity of the Stokes matrices and the fact that they have 1's on the main diagonal. Moreover, by using the symmetries of $A(\lambda, x)$ as we did before \eqref{connection matrix 2}, one gets that Stokes matrices must satisfy
\begin{itemize}
    \item anti-symmetry: $S_1^{(\infty)} = \sigma_3 \left[S_2^{(\infty)}\right]^{T-1} \sigma_3$;
    \vspace{.1cm}
    \item inversion symmetry: $S_n^{(0)} = \sigma_3 S_{n+1}^{(\infty)} \sigma_3, \quad n \in \Z$;
    \vspace{.1cm}
    \item cyclic symmetry: $S_1^{(\infty)} = \sigma_1 S_2^{(\infty)} \sigma_1$;
    \vspace{.1cm}
    \item reality: $\overline{S_1^{(\infty)}} = \sigma_1 \left[ S_2^{(\infty)} \right]^{-1} \sigma_1$;
\end{itemize}
 and therefore they depend on a single real parameter (Stokes parameter) $s^{\R} \in \R$:
\begin{align}
\begin{aligned}
    &S^{(\infty)}_1 = \left[ S_2^{(0)}\right]^{-1} = \begin{pmatrix}
        1 & a\\
        0 & 1
    \end{pmatrix}, \quad S^{(\infty)}_2 = \left[S_1^{(0)}\right]^{-1} = \begin{pmatrix}
        1 & 0\\
        a & 1
    \end{pmatrix},\\
    &a = i s^{\R} \in i \R.
\end{aligned} \label{Stokes matrices}
\end{align}
Using the same symmetries once more, one also gets that
\begin{align}
\sigma_1 E \sigma_1 = S_2^{(0)} E S_1^{(\infty)} = \overline E.
\label{relations for E}
\end{align}
Because \( \det E =1\), relations \eqref{relations for E} show that
\begin{align}
    E = \begin{pmatrix}
        A & B^{\R}\\
        B^{\R} & \overline{A}
    \end{pmatrix}, \quad |A|^2 = 1 + (B^\R)^2, \quad A - \overline{A} = a B^\R, \quad B^\R\in\R. \label{E parametrization}
\end{align}
If we write $a = i s^{\R}$ and $A = x^{\R} + i y^{\R}$ where $x^{\R}, y^{\R} \in \R$, expression \eqref{E parametrization} implies that $E$ is parametrized by triples $(x^{\R}, s^{\R}, B^{\R})$ such that
\begin{align}
    (x^{\R})^2 + \left( \frac{(s^{\R})^2}{4} - 1\right)(B^{\R})^2 = 1 \label{monodromy surface},
\end{align}
see Figure~\ref{monodromy surface pic}. We shall call the set of the above triples the monodromy surface 
\begin{align}
    \mathcal{M} = \left\{ m \equiv (x^{\R}, s^{\R}, B^{\R}) \,|\, \text{$m$ satisfies }\eqref{monodromy surface} \right\}.
\end{align}
We refer those who are interested in the geometry of $\mathcal{M}$ in detail to the monograph \cite{GueHert}.
\begin{figure}[htbp]
    \centering
    \includegraphics[scale = .5]{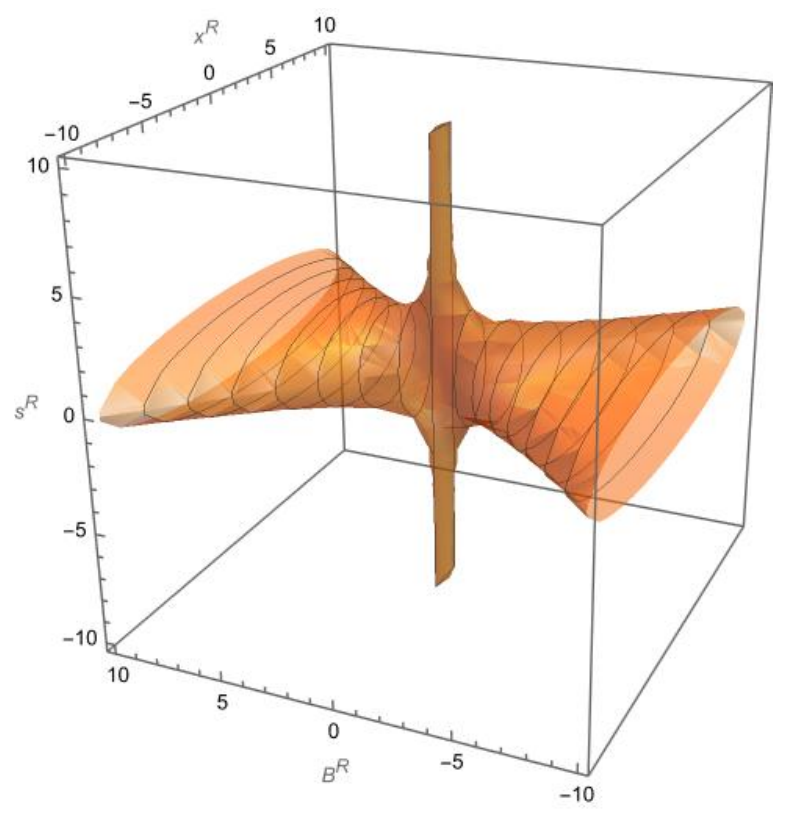}
    \caption{Monodromy surface $\mathcal{M}$ for the sinh-Gordon Painlev\'e III equation.}
    \label{monodromy surface pic}
\end{figure}
What we have shown is that one can always associate the monodromy data $m \in \mathcal{M}$ to each solution $u(x)$ of \eqref{sinh-gordon PIII}.

Another way to parameterize the monodromy data is as follows. Let \( p\in\overline\C\setminus\overline{\mathbb D}\) and \( \iota\in\{\pm1\} \). Set
\begin{equation}
\label{p parametrization}
\begin{cases}
\displaystyle s^\R = -2\Im(p), \quad B^\R = \frac{\iota}{\sqrt{|p|^2-1}}, \quad A = \overline p B^\R, & p\neq\infty, \smallskip \\
s^\R \text{ is any}, \quad B^\R =0, \quad A = \iota, & p=\infty. 
\end{cases}
\end{equation}
It can be readily verified that conditions \eqref{E parametrization} and \eqref{monodromy surface} are indeed fulfilled. Conversely, given a triple \((x^\R,s^\R,B^\R)\), we can set \( p= \overline A/B^\R \), \( \iota = \sign(B^\R)\) when \( B^\R \) is non-zero and \( p=\infty \), \( \iota = A \) otherwise.

\subsection{Riemann-Hilbert Problem for sinh-Gordon Equation}\label{2.2}

Riemann-Hilbert problems arise when one considers the inverse monodromy map
\begin{align*}
    \mathcal{M} \ni m \mapsto u(x, m).
\end{align*}
Given any \( \rho>0 \),  define a sectionally holomorphic in \( \C^* \) matrix function
\begin{equation}
\label{define Psi hat}
    \begin{cases}
    \Psi_1^{(0)}(\lambda), & \arg(\lambda)\in(-\frac\pi2,\frac\pi2),~~0<|\lambda|<\rho, \\
    \Psi_2^{(0)}(\lambda), & \arg(\lambda)\in(-\frac{3\pi}2,-\frac\pi2),~~0<|\lambda|<\rho, \\
    \Psi_1^{(\infty)}(\lambda), & \arg(\lambda)\in(-\frac\pi2,\frac\pi2),~~|\lambda|>\rho, \\
    \Psi_2^{(\infty)}(\lambda), & \arg(\lambda)\in(\frac\pi2,\frac{3\pi}2),~~|\lambda|>\rho.
\end{cases}
\end{equation}
It readily follows from \eqref{connection matrix}--\eqref{defn of stokes mat} that \eqref{define Psi hat} has constant jumps on the circle $S_\rho^1:=\{|\lambda|=\rho\}$ and the imaginary axis as specified on Figure~\ref{Psi hat problem pic} (here, one needs to recall \eqref{riemann surface} to find the jumps on \((0,i\rho)\) and \((-i\infty,-i\rho)\)) as well as asymptotic behavior at $\lambda = 0, \infty$ given by the formal solutions $\Psi_f^{(0, \infty)}$. By construction, the jump matrices depend only on the monodromy data.

Conversely, let \( m\in\mathcal M\) be monodromy data and \( S_{1,2}^{(0,\infty)}\) and \( E \) be as in \eqref{Stokes matrices} and \eqref{relations for E} corresponding to \( m \). Consider the following Riemann-Hilbert problem.

\begin{RHP} \label{rhp original}
Find a $2 \times 2$ matrix function $\hat{\Psi}(\lambda)$ such that
\begin{enumerate}
    \item \text{sectional analyticity}: $\hat{\Psi}(\lambda)$ is analytic for $\lambda \in \C \setminus \Gamma$, where the oriented contour $\Gamma$ is depicted in Figure \ref{Psi hat problem pic};
    \item \text{jump condition}: one-sides traces $\hat{\Psi}_{\pm}(\lambda)$ defined by
    \begin{align*}
        \hat{\Psi}_{\pm}(\lambda) := \lim_{\{ \pm \text{ side of } \Gamma\} \ni \lambda' \to \lambda \in \Gamma } \hat{\Psi}(\lambda'),
    \end{align*}
    where a subscript $+$ (resp. $-$) refers to a limit to the oriented contour from its left (resp. right) side,    
    exists a.e. on \( \Gamma \), belong to \( L^2(\Gamma) \), and satisfy
    \begin{align*}
        \hat{\Psi}_+(\lambda) = \hat{\Psi}_-(\lambda) G_{\hat{\Psi}}, \quad \lambda \in \Gamma,
    \end{align*}
    where the jump matrices $G_{\hat{\Psi}}$ on \( \Gamma \) are as on Figure \ref{Psi hat problem pic};
    \item \text{normalization conditions}: it holds that
    \begin{align}
    \label{Psi hat as lambda -> 0}
        \hat{\Psi}(\lambda) = \begin{cases}
            (I + \O(1/\lambda)) e^{\frac{i x^2 \lambda}{16} \sigma_3} &  \text{as} \quad \lambda\to\infty, \smallskip \\
            P_0 (I + \O(\lambda)) e^{-\frac{i}{\lambda} \sigma_3} & \text{as} \quad \lambda\to0.
        \end{cases}
    \end{align}
\end{enumerate}
\end{RHP}
\begin{figure}[htbp]
    \centering
    \includegraphics[width=7cm]{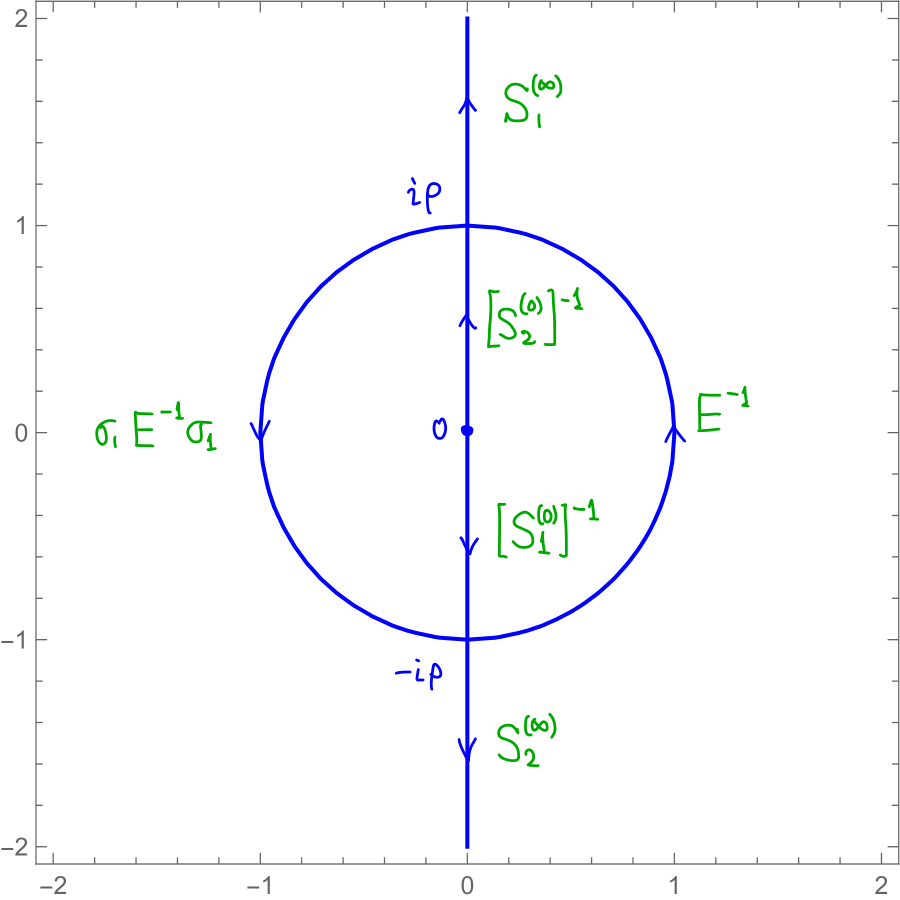}
    \caption{The contour $\Gamma$ and the jump matrices $G_{\hat{\Psi}}$ for RHP~\ref{rhp original}.}
    \label{Psi hat problem pic}
\end{figure}

It is a non-trivial but standard result in integrable systems, see for example \cite[Theorem~5.4]{FIKN} and \cite{ItsNiles}, that the solution of the above problem exists and is meromorphic in \( x \), determines solutions of \eqref{Lax pair sinh-Gordob PIII} via \eqref{define Psi hat}, and allows one to recover the corresponding solution of \eqref{sinh-gordon PIII} via asymptotics around the origin. The proof of Theorem~\ref{main thm} proceeds by asymptotically analyzing RHP~\ref{rhp original} for large \( x \).

\begin{rem}
The choice of \( \iota \), see \eqref{p parametrization}, does not affect real solutions of \eqref{sinh-gordon PIII}. Indeed, let \( u_+(x) \) be the solution of \eqref{sinh-gordon PIII} corresponding to some \( p \) and \( \iota = 1 \) and \( \hat\Psi \) be the corresponding solution of RHP~\ref{rhp original}. Define
\begin{align*}
    \Tilde{\Psi} = \begin{dcases}
        -\hat{\Psi}, & |\lambda| < \rho,\\
        \hat{\Psi}, &|\lambda| > \rho.
    \end{dcases}
\end{align*}
One can readily see that \( \Tilde\Psi \) satisfies RHP~\ref{rhp original}(2) with jump matrices corresponding to the same \( p \) and \( \iota = -1 \). Because \( -P_0(u_+) = P_0(u_++2\pi i)\) and we define real solutions as solutions that are real modulo addition of integer multiples of \( 2\pi i \), \( \Tilde\Psi \) is in fact the solution of RHP~\ref{rhp original} corresponding to \( p \) and \( \iota = -1 \). Thus, \( u_-(x) \), the solution of \eqref{sinh-gordon PIII} corresponding to \( p \) and \( \iota = -1 \), is equal to \( u_+(x) \).
\end{rem}

\subsection{Riemann-Hilbert Problem for sine-Gordon Equation}

The direct monodromy problem for the sine-Gordon equation is well known, see for example \cite[Chapter 13]{FIKN} or \cite{Niles}, and can be deduced in the same way as above. Following \cite{Niles}, the monodromy data for all (not necessarily real) solutions of \eqref{sine-gordon PIII} consist of Stokes matrices
\begin{equation}
\label{sine stokes data}
S_{1,\sin}^{(\infty)} = S_{2,\sin}^{(0)} = \begin{pmatrix} 1 & 0 \\ \varsigma & 1 \end{pmatrix} \quad \text{and} \quad S_{2,\sin}^{(\infty)} = S_{1,\sin}^{(0)} = \begin{pmatrix} 1 & \varsigma \\ 0 & 1 \end{pmatrix}
\end{equation}
as well as the connection matrix
\begin{equation}
\label{sine connection matrix}
E_{\sin} = \pm i \sigma_1 \quad \text{or} \quad E_{\sin} = \frac{\pm1}{\sqrt{1+p_{\sin} q_{\sin}}} \begin{pmatrix} 1 & p_{\sin} \\ -q_{\sin} & 1 \end{pmatrix}, \quad \begin{cases}
    1+p_{\sin} q_{\sin} \neq 0, \\ \varsigma = p_{\sin} + q_{\sin},
\end{cases}
\end{equation}
where the first case is known as {\it special case} and the second one as the {\it general case}. Again, exactly as in the case of the sinh-Gordon equation, the inverse monodromy problem relies on the following Riemann-Hilbert problem.

\begin{RHP} \label{rhp niles}
Find a $2 \times 2$ matrix function $\hat{\Phi}(\lambda)$ such that
\begin{enumerate}
    \item $\hat{\Phi}(\lambda)$ is analytic for $\lambda \in \C \setminus \Gamma$, where $\Gamma$ is the same contour as in RHP~\ref{rhp original};
    \item one-sides traces $\hat{\Phi}_{\pm}(\lambda)$ exists a.e. on \( \Gamma \), belong to \( L^2(\Gamma) \), and satisfy
    \begin{align*}
        \hat{\Phi}_+(\lambda) = \hat{\Phi}_-(\lambda) G_{\hat{\Phi}}, \quad \lambda \in \Gamma,
    \end{align*}
    where $G_{\hat{\Phi}}$ on \( \Gamma \) is as on Figure \ref{Phi hat problem pic} and in \eqref{sine stokes data} and \eqref{sine connection matrix};
    \item it holds that
    \begin{align}
    \label{Phi hat as lambda -> 0}
            \hat{\Phi}(\lambda) = \begin{cases}
            (I + \O(1/\lambda)) e^{\frac{-i x^2 \lambda}{16} \sigma_3} &  \text{as} \quad \lambda\to\infty, \smallskip \\
            P_{\sin} (I + \O(\lambda)) e^{-\frac{i}{\lambda} \sigma_3} & \text{as} \quad \lambda\to0,
        \end{cases}
    \end{align}        
where \(\displaystyle P_{\sin} = \begin{pmatrix}
            \cos \frac{v}{2} & -i \sin \frac{v}{2} \\
            -i \sin \frac{v}{2} & \cos \frac{v}{2}
        \end{pmatrix} \).
\end{enumerate}
\end{RHP}

As usual, \eqref{Phi hat as lambda -> 0} describes asymptotics of the formal fundamental solutions of the first equation in \eqref{Niles pair} at \( \lambda=0,\infty \). 

\begin{figure}[htbp]
    \centering
    \includegraphics[width=7cm]{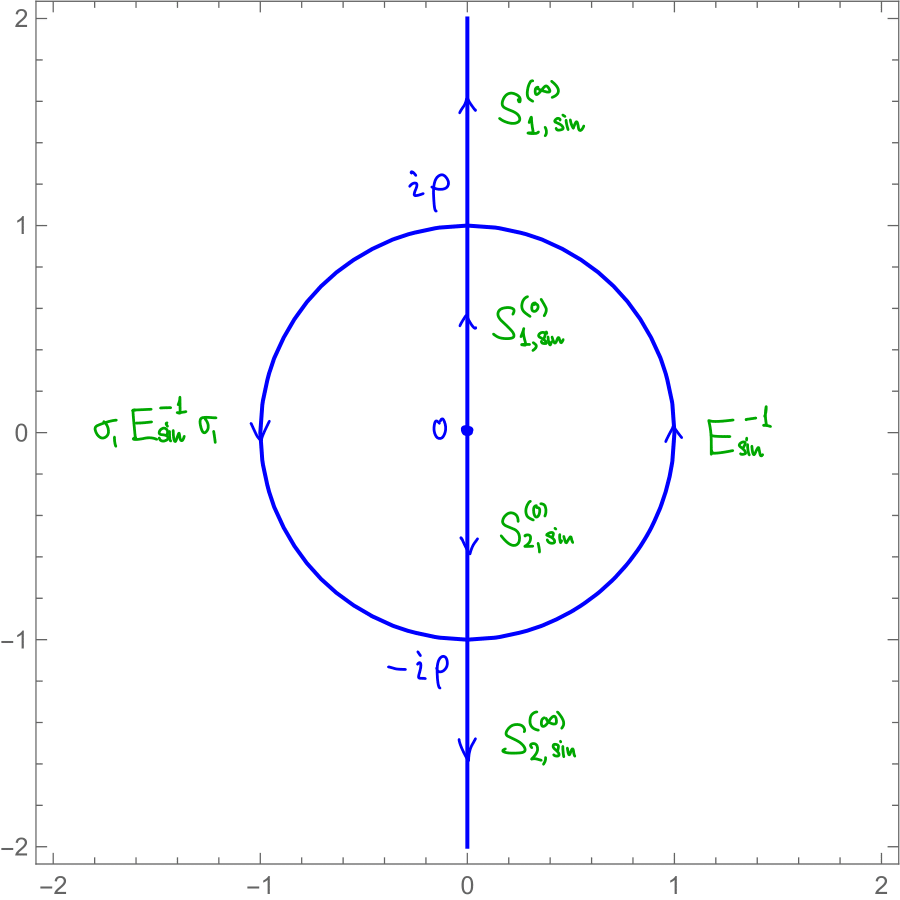}
    \caption{The jump matrices $G_{\hat{\Phi}}$ for RHP~\ref{rhp niles}.}
    \label{Phi hat problem pic}
\end{figure}

Recall that Lax pair \eqref{Niles pair} can be transformed into \eqref{Lax pair sinh-Gordob PIII} by setting either $\Psi (\lambda) = i \sigma_1 \Phi (\lambda)$ or $\Psi (\lambda) = \sigma_1 \Phi (\lambda) \sigma_1$ and making a change of variables $u(x) = i \pi - i v(x)$. Notice that the first gauge transformation takes the formal fundamental solution of \eqref{Niles pair} around \( \lambda=0\) into the formal fundamental solutions of \eqref{Lax pair sinh-Gordob PIII} around \( \lambda=0\) while the second one provides correspondence between the formal fundamental solutions at \( \lambda=\infty \). Therefore, RHP~\ref{rhp niles} transforms into RHP~\ref{rhp original} upon setting
\begin{align}
    \hat{\Psi}(\lambda) = 
    \begin{dcases}
        \sigma_1 \hat{\Phi}(\lambda) \sigma_1, & |\lambda| > \rho, \\
        i \sigma_1 \hat{\Phi}(\lambda), & |\lambda| < \rho.
    \end{dcases} \label{phi-hat to psi-hat}
\end{align}
The above transformation induces the following relation between monodromy data:
\[
\begin{cases}
    S_{1}^{(\infty)} = \sigma_1 S_{1,\sin}^{(\infty)} \sigma_1, \quad S_{2}^{(\infty)} = \sigma_1 S_{2,\sin}^{(\infty)} \sigma_1, \smallskip \\
    S_{2}^{(0)} = S_{1,\sin}^{(0)-1}, \quad S_{1}^{(0)} = S_{2,\sin}^{(0)-1}, \quad E = -i E_{\sin}\sigma_1.
\end{cases}
\]
Using \( (\iota,p) \) parametrization \eqref{p parametrization} of the monodromy data for sinh-Gordon equation, we get that \( p=\infty \) corresponds to the special case with \( \varsigma=a=is^\R \) and \( 1<|p|<\infty\) corresponds to the general case with \( p_{\sin} = \overline p \) and \( q_{\sin} = -p\) where again \(\varsigma = p_{\sin} + q_{\sin} = -2i\Im(p) = i s^\R \). In both cases \( \iota \) corresponds to the choice of sign in \eqref{sine connection matrix} and we interpret \( 1/\sqrt{1+p_{\sin}q_{\sin}} \) as \( i |B^\R| = i/\sqrt{|p|^2-1}\). The choices of \( p_{\sin},q_{\sin} \) that cannot be parametrized by \( p \) as above lead to the non-real solutions \( u(x) \) of sinh-Gordon equations.

The results of Theorem~\ref{Niles result} now follow from \cite[Theorems~7-9]{Niles} and the above correspondence for the monodromy data, see also Appendix~\ref{connection P3}.

\section{Proof of Theorem \ref{main thm}}

As indicated in the introduction, we shall use the nonlinear steepest descent method of Deift and Zhou \cite{DZ} to prove the desired asymptotics. In Section~\ref{case 1} we review the case $B^{\R} = 0$ that includes the McCoy-Tracy-Wu one-parameter family of smooth solutions (this proof has already been outlined in Niles' thesis \cite[Apendix~D]{Niles}). Then, in Section~\ref{case 2}, we consider the case $B^{\R} \neq 0$.

\subsection{Case 1: $B^{\R} = 0$} \label{case 1}

When $B^{\R} = 0$, it readily follows from \eqref{monodromy surface} that $A = \pm 1$. As we pointed out in the remark of Section~\ref{2.2}, it is enough to consider the case $A = 1$ only. Then, the jump matrices on the arcs of \( S^1_\rho \) become equal to the identity matrix.  That is, $\hat\Psi$ has the discontinuity only along the imaginary line. Let
\begin{align}
    Y(\lambda) := \hat{\Psi}\left( \frac{4\lambda}{x} \right) e^{-x \theta(\lambda)\sigma_3}, \quad \theta(\lambda) := \frac{i}{4}\left(  \lambda - \frac1\lambda \right).
    \label{defn of Y}
\end{align}
Then \( Y(\lambda) \) is the solution of the following Riemann-Hilbert problem: 
\begin{RHP}
\label{RHP Y}
Find a $2 \times 2$ matrix function $Y(\lambda)$ such that
\begin{enumerate}
    \item $Y(\lambda)$ is analytic for $\lambda \in \C \setminus \Gamma_Y$, where \( \Gamma_Y\) is the imaginary line oriented as on Figure~\ref{gl Y problem pic};
    \item one-sides traces $Y_\pm(\lambda)$ exists a.e. on \( \Gamma_Y \), belong to \( L^2(\Gamma_Y) \), and satisfy
    \begin{align*}
        Y_+(\lambda) = Y_-(\lambda) G_{Y}, \quad \lambda \in \Gamma_Y,
    \end{align*}
    where the jump matrices $G_{Y}$ are given by
    \begin{align}
\label{defn GY}
    G_{Y} = e^{x \theta (\lambda) \sigma_3} \, G_{\hat{\Psi}} \, e^{-x \theta (\lambda) \sigma_3}, \quad \lambda \in \Gamma_Y;
\end{align}
    \item it holds that
    \begin{align*}
        Y(\lambda) = \begin{cases}
            I + \O(1/\lambda)  &  \text{as} \quad \lambda\to\infty, \smallskip \\
            P_0 (I + \O(\lambda)) & \text{as} \quad \lambda\to0,
        \end{cases}
    \end{align*}
    where \( P_0 \) was defined in \eqref{defn of P0}.
\end{enumerate}
\end{RHP}

\begin{figure}[ht!]
    \centering
    \includegraphics[width=6cm]{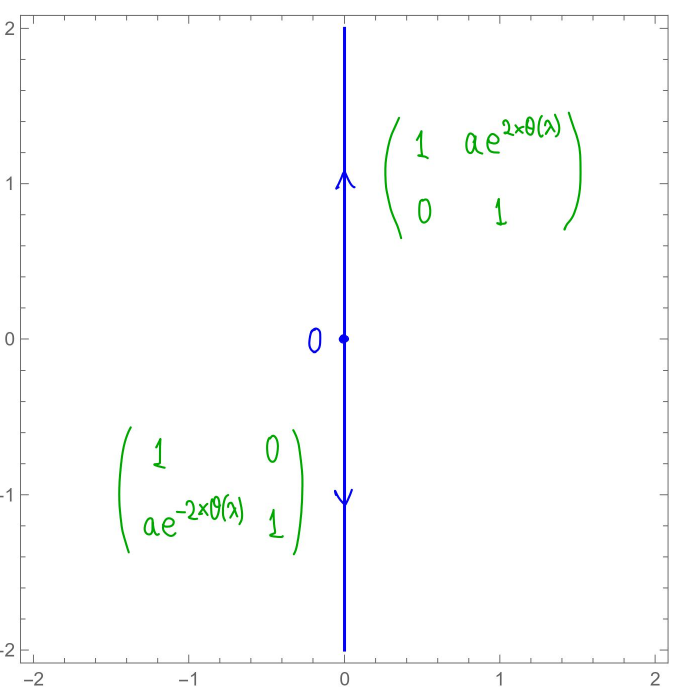}
    \caption{The jump of $Y$ on the imaginary line \( \Gamma_Y \).}
    \label{gl Y problem pic}
\end{figure}

Note that $2 \theta(\lambda) < -1$ when \( \lambda\in\Gamma_Y \), \(\Im(\lambda)>0\), and $-2 \theta(\lambda) <-1$ when \( \lambda\in\Gamma_Y \), \(\Im(\lambda)<0\). Thus, both jump matrices comprising $G_Y$ exponentially decay to the identity matrix $I$. In fact, one can show that
\begin{align}
    \| G_Y - I \|_{L^2(\Gamma_Y) \cap L^{\infty} (\Gamma_Y)} \leq C \, \frac{e^{-x}}{x^{1/4}},
    \label{gl G-I is small}
\end{align}
where $C$ is some constant. Hence, by the small norm theorem, see for example \cite[Theorem~8.1]{FIKN}, there exists a unique solution of RHP~\ref{RHP Y} for all $x$ large and
\begin{align}
    Y(\lambda) = I + \frac{1}{2 \pi i} \int_{\Gamma_Y} \frac{\rho(\lambda')(G_{Y}(\lambda') - I)}{\lambda' - \lambda} d\lambda', \quad \lambda \notin \Gamma_{Y}, \label{gl Y singular int eq}
\end{align}
where $\ds \rho(\lambda') = Y_-(\lambda^\prime)$ for $\lambda' \in \Gamma_Y$ (hence, \( \rho(\lambda^\prime) \) is the solution of the singular integral equation obtained by taking the traces of both sides of \eqref{gl Y singular int eq} on the left-hand side of \(\Gamma_Y\)). The small norm theorem further implies that 
\begin{align}
    \|\rho - I\|_{L^2(\Gamma_Y)} \leq \Tilde{C} \, \frac{e^{-x}}{x^{1/4}} \label{gl rho - I has a small norm}
\end{align}
for some constant $\Tilde{C}$. Observe that \( G_{Y}(\lambda') - I \) vanishes exponentially at the origin. Hence, \( Y_+(0) = Y_-(0) =: Y(0)\) and one has that
\begin{align}
\begin{aligned}
    Y(0) &= I + \frac{1}{2\pi i} \int_{\Gamma_{Y}} \frac{G_{Y}(\lambda') - I}{\lambda'} d\lambda'  + \frac{1}{2\pi i} \int_{\Gamma_Y} \frac{(\rho(\lambda') - I)(G_{Y}(\lambda') - I)}{\lambda'} d\lambda'\\
    &= I + \frac{1}{2\pi i} \int_{\Gamma_{Y}} \frac{G_{Y}(\lambda') - I}{\lambda'} d\lambda' + \O \left(\frac{e^{-2x}}{\sqrt{x}}\right),
\end{aligned} \label{gl Y asymptotic eq 1}
\end{align}
by Cauchy-Schwarz inequality, \eqref{gl G-I is small}, and \eqref{gl rho - I has a small norm}. Recall that
\[
    \frac{1}{2 \pi} \int_{0}^{\infty} e^{-\frac{x}{2}(y + \frac{1}{y})} \frac{dy}{y} = \frac1{2\pi} \int_{-\infty}^\infty e^{-x\cosh s}ds = \frac1\pi K_0(x),
\]
see \cite[Equation (10.32.9)]{NIST:DLMF}. Then, \eqref{gl Y asymptotic eq 1} becomes
\begin{align}
    Y(0) = \begin{pmatrix}
    1 & \frac1\pi s^\R K_0(x)\\
    \frac1\pi s^\R K_0(x) & 1
    \end{pmatrix} + \O \left(\frac{e^{-2x}}{\sqrt{x}}\right). \label{gl Y asymptotic eq 2}
\end{align}
On the other hand, it immediately follows from RHP~\ref{RHP Y}(3) and \eqref{defn of P0} that
\begin{align}
    Y(0) = \begin{pmatrix}
        \cosh \frac{u}{2} & \sinh \frac{u}{2}\\
        \sinh \frac{u}{2} & \cosh \frac{u}{2}
    \end{pmatrix}. \label{gl Y asymptotic eq 3}
\end{align}
By comparing \eqref{gl Y asymptotic eq 2} with \eqref{gl Y asymptotic eq 3}, we get that
\[
u(x) = 2\ln \left[ 1 + \frac{s^\R}\pi K_0(x) + \O\left(\frac{e^{-2x}}{\sqrt{x}}\right) \right]
    = \frac{2 s^\R}\pi K_0(x) + \O\left(\frac{e^{-2x}}{\sqrt{x}}\right)
\]
as $x \to \infty$, where we used the fact that \( 0\leq \sqrt x e^x K_0(x) \leq C^\prime \) to deduce the leading order behavior of the logarithm. 
\begin{comment}
Observe that the substitutions \( y=t-\sqrt{t^2-1} \) when \( y<1 \) and \( y=t+\sqrt{t^2-1} \) when \(y>1 \) give
\[
I(x) = \frac1\pi\int_1^\infty e^{-xt}\frac{dt}{\sqrt{t^2-1}} = \frac{e^{-x}}\pi\int_0^\infty e^{-xt}\frac{dt}{\sqrt{t(t+2)}}.
\]
By using the Maclaurin series of \( 1/\sqrt{t+2}\), we obtain that
\[
I(x) = \frac{e^{-x}}{\sqrt {2\pi}} \sum_{n=0}^N \frac{(-1)^n}{\pi2^n}\frac{\Gamma(n+1/2)}{\Gamma(n+1)} \int_0^\infty t^{n-1/2}e^{-xt}dt + \O_N \left( e^{-x}\int_0^1 t^{N+1/2}e^{-xt}dt \right)
\]
for any non-negative integer \( N \). Recalling the definition of the Gamma function, or equivalently, the Laplace transform of the power functions,
\end{comment}
This finishes the proof of Theorem~\ref{main thm} in the case \( B^\R=0 \).

\begin{comment}
\begin{rem}
Strictly speaking, this derivation of the large $x$ asymptotics of the one-parameter family of solutions will provide a little broader class of solutions, i.e., solutions whose asymptotics are
\begin{align}
    u(x) = 4 \pi k i + \sqrt{\frac{2}{\pi}} \, s^{\R} \, \frac{e^{-x}}{\sqrt{x}} + \O\left(\frac{e^{-2x}}{\sqrt{x}} \right), \quad x\to \infty, \label{gl cplx ans}
\end{align}
for $k \in \Z$. Moreover, the case $A=-1$ provides us with the extra shift of \eqref{gl cplx ans} by $2\pi i n$ where $n \in \Z$. However, as we consider the real-valued solutions of \eqref{sinh-gordon PIII}, we ignore all the integer multiples of $2 \pi i$ and obtain the large $x$ asymptotics \eqref{gl ans}.
    
Moreover, if we consider the case $A = -1$, one can proceed with almost the same proof as we presented here and obtain
\begin{align}
    u(x) = 4 \pi k i + 2 \pi i + \sqrt{\frac{2}{\pi}} \, s^{\R} \, \frac{e^{-x}}{\sqrt{x}} + \O\left(\frac{e^{-2x}}{\sqrt{x}} \right), \quad x\to \infty,
\end{align}
for $k \in \Z$.
These observations are just because the $\hat{\Psi}$-RH problem is not distinguished with the $2 \pi k i$ shift from the real-valued solutions of \eqref{sinh-gordon PIII}.
\end{rem}
\end{comment}

%%%%%%%%%%%%%%%%%%%%%%%%%%%%%%%%%

\subsection{Case 2: $B^{\R} \neq 0$}
\label{case 2}

Now we will consider the case $B^{\R} \neq 0$. Again, we only look at \( B^\R>0 \).

\subsubsection{Opening of Lenses}

Recall that \( \rho \) in RHP~\ref{rhp original} was arbitrary. Motivated by the scaling used in \eqref{defn of Y}, we now set \( \rho = 4/x \). Transformation \eqref{defn of Y} no longer fits our needs because $\Re( \theta(\lambda)) < 0$ in the upper half-plane while $\Re(\theta(\lambda)) > 0$ in the lower half-plane and so the jump matrices on the unit circle would always have exponentially growing entries. The next best thing one can do is to make off-diagonal entries of the jump matrices on the unit circle oscillating. To this end, consider
\begin{align}
\label{defn X}
    X(\lambda) := \hat{\Psi}\left( \frac{4\lambda}{x} \right) \begin{Bmatrix}
        I, & |\lambda| > 1 \\
        \sigma_1, & |\lambda| < 1
    \end{Bmatrix} e^{-x \varphi(\lambda)\sigma_3}, \quad \varphi(\lambda) = \frac{i}4\left(\lambda + \frac{1}{\lambda} \right).
\end{align}
Then \( X(\lambda) \) is the solution of the following Riemann-Hilbert problem: 
\begin{RHP}
\label{RHP X}
Find a $2 \times 2$ matrix function $X(\lambda)$ such that
\begin{enumerate}
\item $X(\lambda)$ is analytic for $\lambda \in \C \setminus \Gamma_X$, where \( \Gamma_X \) is the union of the imaginary axis and the unit circle oriented as on Figure~\ref{X problem pic};
    \item one-sides traces $X_\pm(\lambda)$ exists a.e. on \( \Gamma_X \), belong to \( L^2(\Gamma_X) \), and satisfy
    \begin{align*}
        X_+(\lambda) = X_-(\lambda) G_{X}, \quad \lambda \in \Gamma_X,
    \end{align*}
    where the jump matrices $G_{X}$ on \( \Gamma_X \) are as on Figure~\ref{X problem pic};
    \item it holds that
    \begin{align}
    \label{X asymptotics lambda -> 0}
        X(\lambda) = \begin{cases}
            I + \O(1/\lambda)  &  \text{as} \quad \lambda\to\infty, \smallskip \\
            P_0 \, \sigma_1 (I + \O(\lambda)) & \text{as} \quad \lambda\to0,
        \end{cases}
    \end{align}
    where \( P_0 \) was defined in \eqref{defn of P0}.
\end{enumerate}
\end{RHP}

\begin{figure}[ht!]
    \centering
    \includegraphics[width=7cm]{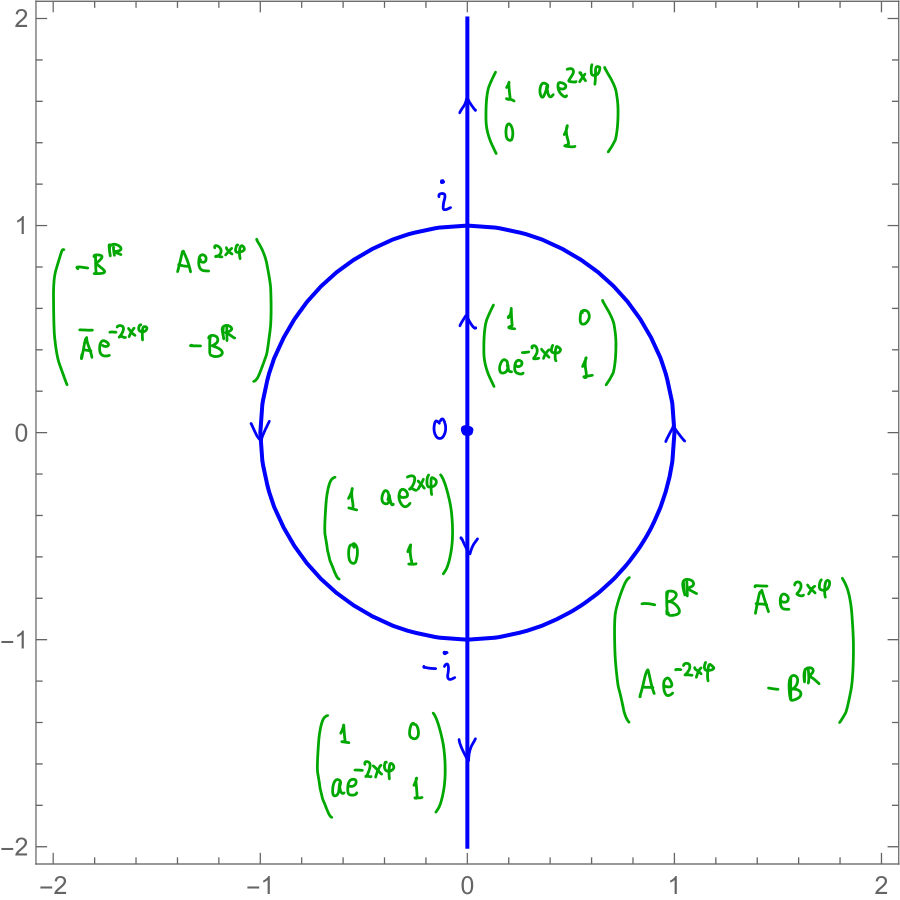}
    \caption{The jump matrices $G_{X}$ for RHP~\ref{RHP X}.}
    \label{X problem pic}
\end{figure}

Let $\lambda = \xi + i \eta$, where $\xi, \, \eta \in \R$. Then,
\begin{align}
\label{real part varphi}
    \Re (2 \varphi (\lambda)) = \frac{\eta \left( 1 - \xi^2 - \eta^2 \right)}{2 \left( \xi^2 + \eta^2 \right)}, \quad \Im (2 \varphi (\lambda)) = \frac{\xi \left( 1 + \xi^2 + \eta^2 \right)}{2 \left( \xi^2 + \eta^2 \right)}.
\end{align}
\begin{figure}[ht!]
\centering
\includegraphics[width=7cm]{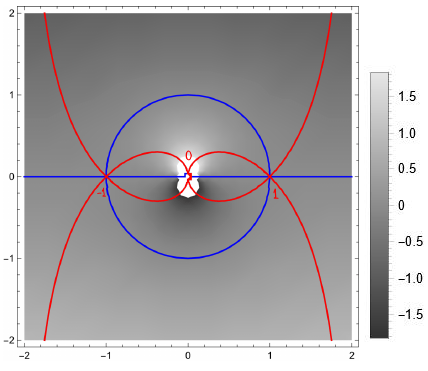}
\caption{The zero level curve $\Re(2\varphi(\lambda)) = 0$ (blue), the sign of $\Re(2\varphi(\lambda))$ (gray), and the stationary contour $\Im(2\varphi) = \pm 1$ (red).}
\label{Re 2 phi}
\end{figure}
Clearly, \( \varphi(\lambda) \) is purely imaginary on the unit circle and therefore off-diagonal entries of \( G_X \) on the unit circle are oscillating as desired. Notice also that the stationary points of $2\varphi(\lambda)$ are $\lambda = \pm 1$ and $2 \varphi ( \pm 1 ) = \pm i$. The stationary contour $\Im (2 \varphi(\lambda)) = \pm 1$ is shown on Figure \ref{Re 2 phi}.

Next, we decompose the jump matrices on the unit circle in RHP~\ref{RHP X}(2) as follows:
\begin{align*}
    &\begin{pmatrix}
        -B^{\R} & \overline{A} e^{2x\varphi}\\
        A e^{-2x\varphi} & -B^{\R}
    \end{pmatrix} = L_1 D_1 R_1 = L_2 D_2 R_2,\\
    &\begin{pmatrix}
        -B^{\R} & A e^{2x\varphi}\\
        \overline{A} e^{-2x\varphi} & -B^{\R}
    \end{pmatrix} = L_3 D_3 R_3 = L_4 D_4 R_4,
\end{align*}
where
\begin{align}
    &L_1 = \begin{pmatrix}
        1 & -\frac{\overline{A}}{B^{\R}} e^{2 x \varphi} \\
        0 & 1
    \end{pmatrix}, \,
    D_1 = \begin{pmatrix}
        \frac{1}{B^{\R}} & 0 \\
        0 & -B^{\R}
    \end{pmatrix}, \,
    R_1 = \begin{pmatrix}
        1 & 0 \\
        -\frac{A}{B^{\R}} e^{-2 x \varphi} & 1
    \end{pmatrix}, \label{L1, D1, R1}\\
    &L_2 = \begin{pmatrix}
        1 & 0 \\
        -\frac{A}{B^{\R}} e^{-2 x \varphi} & 1
    \end{pmatrix}, \,
    D_2 = \begin{pmatrix}
        -B^{\R} & 0 \\
        0 & \frac{1}{B^{\R}}
    \end{pmatrix}, \,
    R_2 = \begin{pmatrix}
        1 & -\frac{\overline{A}}{B^{\R}} e^{2 x \varphi} \\
        0 & 1
    \end{pmatrix}, \label{L2, D2, R2}\\
    &L_3 = \begin{pmatrix}
        1 & -\frac{A}{B^{\R}} e^{2 x \varphi} \\
        0 & 1
    \end{pmatrix}, \,
    D_3 = \begin{pmatrix}
        \frac{1}{B^{\R}} & 0 \\
        0 & -B^{\R}
    \end{pmatrix}, \,
    R_3 = \begin{pmatrix}
        1 & 0 \\
        -\frac{\overline{A}}{B^{\R}} e^{-2 x \varphi} & 1
    \end{pmatrix}, \label{L3, D3, R3}\\
    &L_4 = \begin{pmatrix}
        1 & 0 \\
        -\frac{\overline{A}}{B^{\R}} e^{-2 x \varphi} & 1
    \end{pmatrix}, \,
    D_4 = \begin{pmatrix}
        -B^{\R} & 0 \\
        0 & \frac{1}{B^{\R}}
    \end{pmatrix}, \,
    R_4 = \begin{pmatrix}
        1 & -\frac{A}{B^{\R}} e^{2 x \varphi} \\
        0 & 1
    \end{pmatrix}. \label{L4, D4, R4}
\end{align}
Using the above decomposition we define a matrix function \( \tilde X(\lambda) \) as on Figure~\ref{defn of X tilde}.
\begin{figure}[ht!]
\centering
\includegraphics[width=7cm]{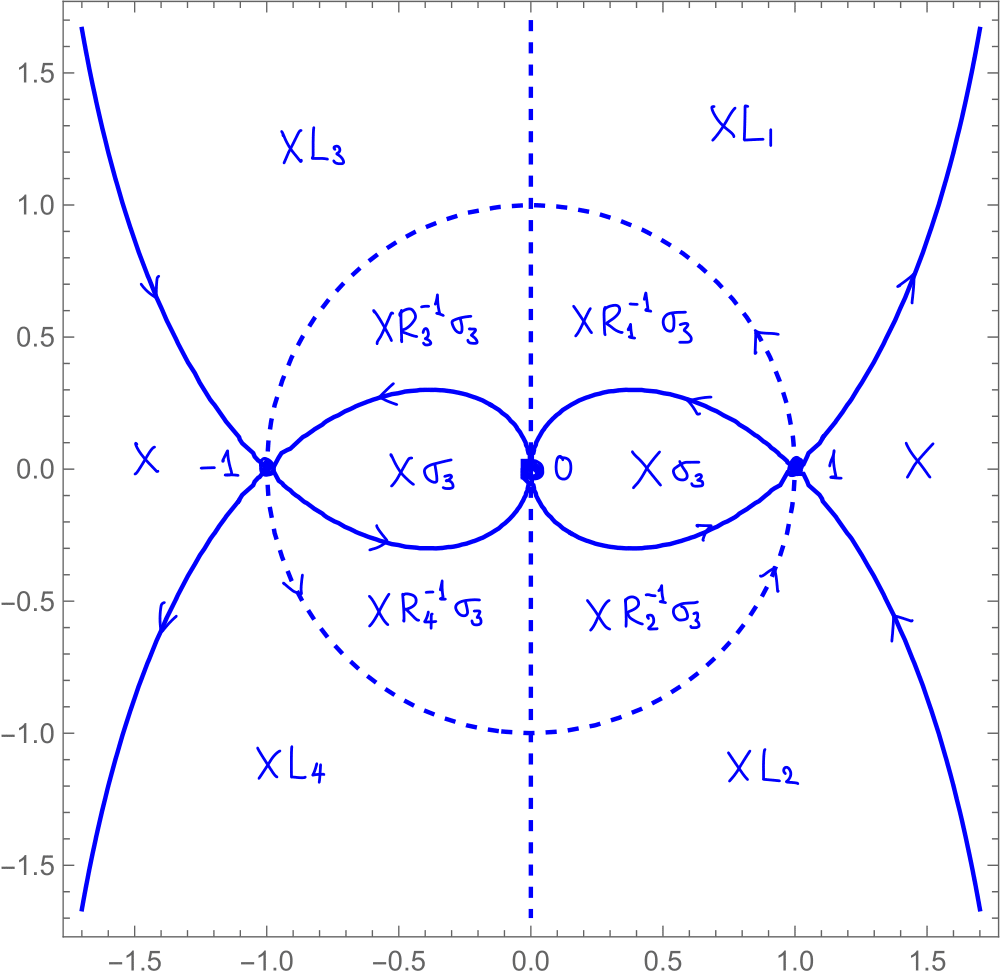}
\caption{Definition of the matrix $\Tilde{X}$. The solid contour is the stationary contour of \( 2\varphi\), see Figure~\ref{Re 2 phi}.}
\label{defn of X tilde}
\end{figure}
Then \( \tilde X(\lambda) \) is the solution of the following Riemann-Hilbert problem: 
\begin{RHP} \label{X tilde RH problem}
Find a $2 \times 2$ matrix function $\Tilde{X}(\lambda)$ such that
\begin{enumerate}
    \item $\Tilde{X}(\lambda)$ is analytic for $\lambda \in \C \setminus \Gamma_{\Tilde{X}}$, where $\Gamma_{\Tilde{X}}$ is the union of the unit circle \( S^1 \) and the stationary contour of \( 2\varphi \) oriented as on Figure~\ref{X tilde problem pic};
    \item the one-sides traces $\tilde X_\pm(\lambda)$ exists a.e. on \( \Gamma_{\tilde X} \), belong to \( L^2(\Gamma_{\tilde X}) \), and satisfy
    \begin{align*}
        \Tilde{X}_+(\lambda) = \Tilde{X}_-(\lambda) G_{\Tilde{X}}, \quad \lambda \in \Gamma_{\Tilde{X}},
    \end{align*}
    where the jump matrices $G_{\Tilde{X}}$ on \( \Gamma_{\tilde X} \) are as in Figure \ref{X tilde problem pic};
    \item it holds that
    \begin{align}
    \label{X tilde near zero}
            \tilde X(\lambda) = \begin{cases}
            I + \O(1/\lambda)  &  \text{as} \quad \lambda\to\infty, \smallskip \\
            P_0 \, \sigma_1\sigma_3 (I + \O(\lambda)) & \text{as} \quad \lambda\to0.
        \end{cases}
    \end{align}
\end{enumerate}
\end{RHP}
\begin{rem}
    Notice that RHP~\ref{X tilde RH problem}(3) readily follows from RHP~\ref{RHP X}(3) upon recalling that \( -\Re(\varphi(\bar\lambda)) = \Re(\varphi(\lambda)) >0 \) for \( \lambda \in \{|\lambda|<1\}\cap\{\Im(\lambda)>0\} \) and \( \lambda \in \{|\lambda|>1\}\cap\{\Im(\lambda)<0\} \), see \eqref{real part varphi}. This observation also implies that the jump matrices \( G_{\tilde X} \) tends to the identity matrix as \( x\to\infty \) on \( \Gamma_{\tilde X}\setminus S^1 \).
\end{rem}
\begin{figure}[htbp]
\centering
\includegraphics[width=7cm]{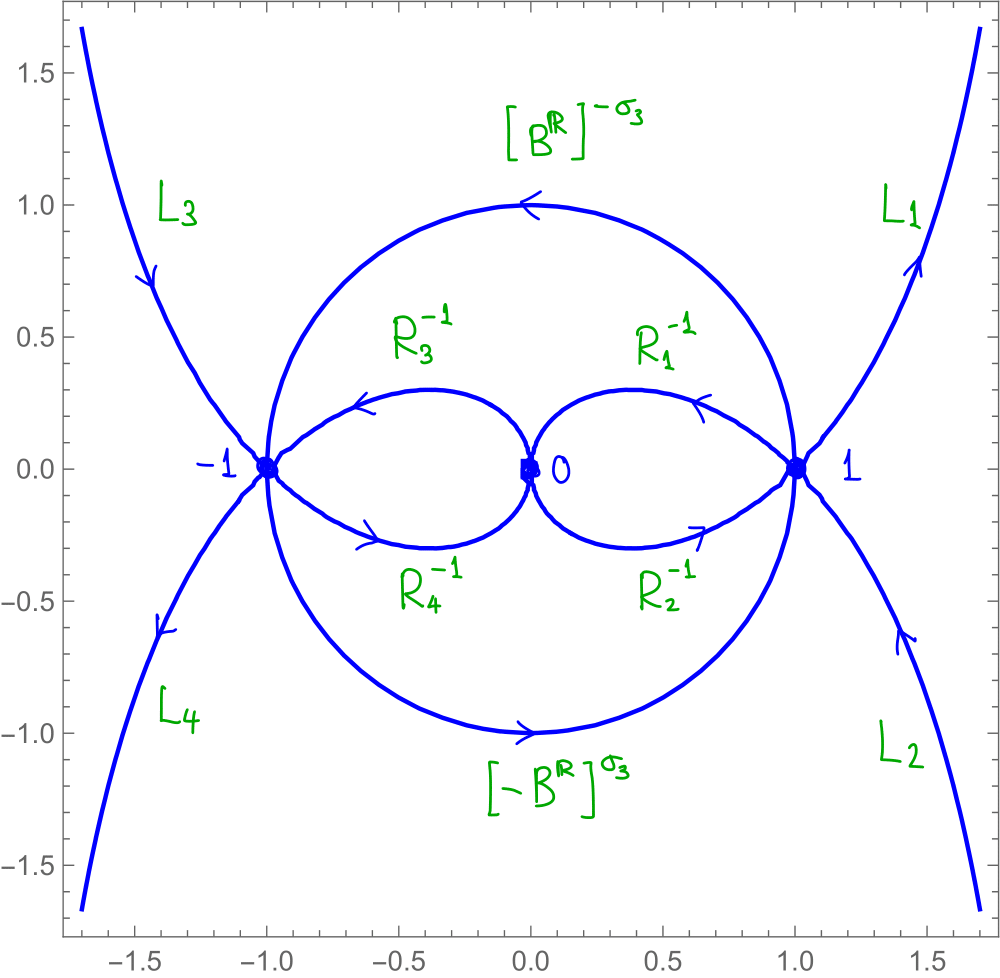}
\caption{The contour $\Gamma_{\Tilde{X}}$ and the jump matrices $G_{\Tilde{X}}$ for RHP~\ref{X tilde RH problem}.}
\label{X tilde problem pic}
\end{figure}

%%%%%%%%%%%%%%%%%%%%%%%%%%%%%%%

\subsubsection{The Global Parametrix} \label{global param}

As we have mentioned above, the jump matrices in RHP~\ref{RHP X} are close to the identity except on the unit circle \( S^1 \) (oriented counter-clockwise). As standard in the theory, we now look for a {\it global parametrix}, a solution of a Riemann-Hilbert problem with the same jump on the unit circle as in RHP~\ref{RHP X}(2). To this end, let
\begin{align*}
    &C_1 := S^1 \cap \{ \lambda \in \C \,|\, \Im \lambda \geq 0\}, \\
    &C_2 := S^1 \cap \{\lambda \in \C \,|\, \Im \lambda \leq 0\}.
\end{align*}
Consider the following model Riemann-Hilbert problem:
\begin{RHP} \label{rhp global}
Find a $2 \times 2$ matrix function $P^{(gl)}(\lambda)$ such that
\begin{enumerate}
    \item $P^{(gl)}(\lambda)$ is analytic for $\lambda \in \overline\C \setminus S^1$;
    \item continuous one-sides traces $P_\pm^{(gl)}(\lambda)$ exists on \( S^1\setminus\{\pm1\} \) and satisfy
    \begin{align*}
        P^{(gl)}_+(\lambda) = P^{(gl)}_-(\lambda) G_{P^{(gl)}}, \quad \lambda \in S^1,
    \end{align*}
    where
    \begin{align*}
        G_{P^{(gl)}} := \begin{dcases}
            \left[B^{\R}\right]^{-\sigma_3} & \text{if} \quad \lambda \in C_1\setminus\{\pm1\},\\
            \left[-B^{\R}\right]^{\sigma_3} & \text{if} \quad \lambda \in C_2\setminus\{\pm1\};
        \end{dcases}
    \end{align*}
    \item it holds that \( P^{(gl)}(\lambda) = I + \O(1/\lambda) \) as $\lambda \to \infty$.
\end{enumerate}
\end{RHP}

\begin{lem}
\label{lemma global par}
    The solution of RHP~\ref{rhp global} is given by
    \begin{align}
        P^{(gl)}(\lambda) = f_1(\lambda) f_2(\lambda),
        \label{defn of Pgl}
    \end{align}
    where
    \[
    f_1(\lambda) := \left( \frac{\lambda -1}{\lambda + 1} \right)^{\frac{1}{2\pi i} (\ln B^{\R}) \sigma_3}
    \]
    is holomorphic in \( \overline \C\setminus C_1 \) and normalized so that \( f_1(0) = \big[ \sqrt{B^\R} \big]^{-\sigma_3} \), and
    \[
    f_2(\lambda) := \left( \frac{\lambda - 1}{\lambda + 1} \right)^{\frac{1}{2\pi i} (\ln B^{\R} + i \pi)\sigma_3}
    \]
    is holomorphic in \( \overline \C\setminus C_2 \) and normalized so that \( f_2(0) = \big[ i\sqrt{B^\R} \big]^{\sigma_3} \). 
\end{lem}

\begin{proof}
We check that $P^{(gl)}(\lambda)$ defined by \eqref{defn of Pgl} satisfies RHP~\ref{rhp global}. 

RHP~\ref{rhp global}(1) follows immediately from the choice of the branch cuts for \( f_1(\lambda) \), \( f_2(\lambda) \). 

Notice that a linear fractional transformation \( (\lambda - 1)/(\lambda + 1) \) maps \( C_1\) and \( C_2\) onto \( \{iy~|~y\geq0\} \) and \( \{iy~|~y\leq0\} \), respectively. Then, one can readily check that
\[
\begin{cases}
f_1(\lambda) & \displaystyle = \exp\left\{\frac{\ln B^{\R}}{2 \pi i}L_{\pi/2} \left( \frac{\lambda - 1}{\lambda + 1} \right)\sigma_3 \right\}, \medskip \\ 
f_2(\lambda) & \displaystyle = \exp\left\{\frac{\ln B^{\R}+\pi i}{2 \pi i}L_{3\pi/2} \left( \frac{\lambda - 1}{\lambda + 1} \right)\sigma_3 \right\},
\end{cases}
\]
where \( L_\theta(z) := \ln|z| + iA_\theta(z) \) and \( A_\theta(z) \) is the argument of \( z \) that belongs to \( (\theta-2\pi,\theta] \) (we use here the fact that \( L_{\pi/2}(-1) = -\pi i \) and \( L_{3\pi/2}(-1) = \pi i \)). Since
\[
L_{\theta+}(z) = L_{\theta-}(z) + 2\pi i
\]
on the ray \( \{z~|~\arg(z)=\theta\} \) oriented towards the origin, it holds that
\[
f_{1+}(\lambda) = f_{1-}(\lambda) \big[B^\R\big]^{\sigma_3} \quad \text{and} \quad f_{2+}(\lambda) = f_{2-}(\lambda) \big[-B^\R\big]^{\sigma_3}
\]
on \( C_1 \) and \( C_2 \), respectively. This immediately shows that RHP~\ref{rhp global}(2) is fulfilled. 

Finally, as \( L_{\pi/2}(1) = L_{3\pi/2}(1) = 0 \), it holds that \( f_1(\infty) = f_1(\infty) = I\), which shows the validity of RHP~\ref{rhp global}(3) and finishes the proof of the lemma.
\end{proof}

Let $K^{(1,2,3)}$ be the following sub-domains of a sufficiently small neighborhood $U_1$ of $\lambda = 1$:
\begin{align*}
    K^{(1)} &= U_1 \cap \text{(outside of the unit circle)}\\
    K^{(2)} &= U_1 \cap \text{(inside of the unit circle)} \cap \text{(upper half plane)}\\
    K^{(3)} &= U_1 \cap \text{(inside of the unit circle)} \cap \text{(lower half plane)}.
\end{align*}
\begin{figure}[H]
    \centering
    \begin{subfigure}{.33\textwidth}
    \includegraphics[width=0.9\linewidth]{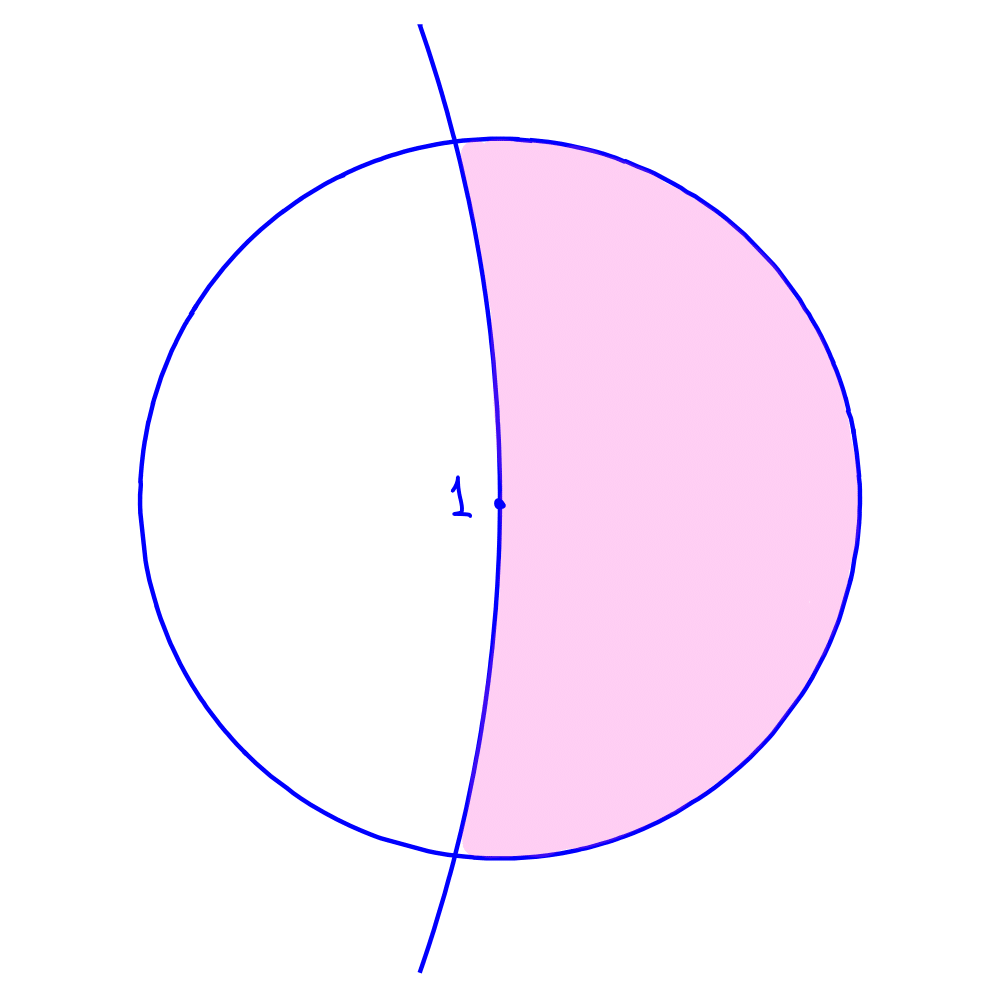}
    \caption{$K^{(1)}$}
    \label{Picture of K^{(1)}}
    \end{subfigure}%
    \begin{subfigure}{.33\textwidth}
    \includegraphics[width=0.9\linewidth]{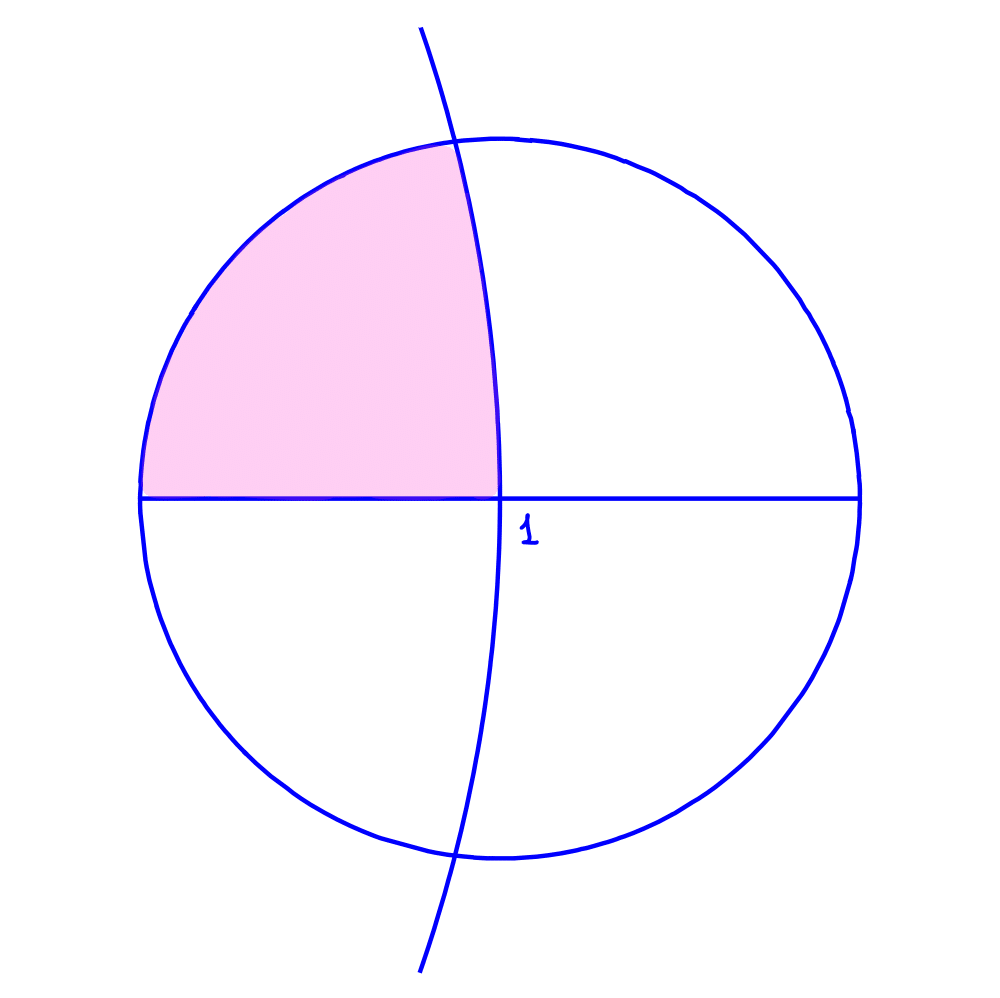}
    \caption{$K^{(2)}$}
    \label{Picture of K^{(2)}}
    \end{subfigure}
    \begin{subfigure}{.33\textwidth}
    \includegraphics[width=0.9\linewidth]{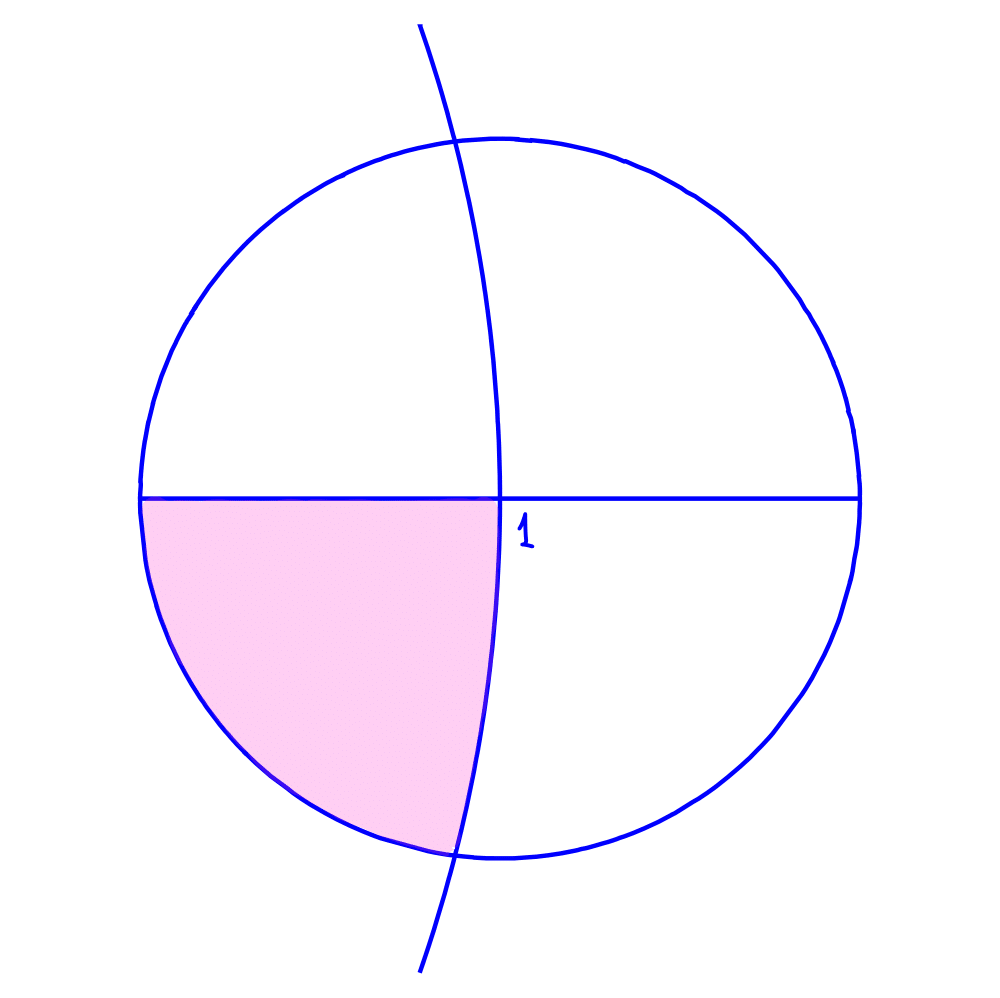}
    \caption{$K^{(3)}$}
    \label{Picture of K^{(3)}}
    \end{subfigure}
    \caption{Regions $K^{(1,2,3)}$.}
    \label{Pictures of K^{(i)}}
\end{figure}

Let \( (\lambda-1)^\nu \) be the branch with the cut \( \{ \lambda \leq 1\} \) positive for \(\lambda>1\), where
\begin{align}
    \nu = \frac{1}{\pi i} \ln B^{\R} + \frac{1}{2} \label{defn of nu}.
\end{align}
Then, one can write the global parametrix $P^{(gl)}(\lambda)$  in \( U_{1} \) as
\begin{equation}
\begin{aligned}
P^{(gl)}(\lambda) = \Theta_1(\lambda)^{-\sigma_3} (\lambda - 1)^{\nu \sigma_3}C^{(1)}, \quad C^{(1)} :=
\begin{dcases}
    I & \text{in } K^{(1)},\\
    \left[ B^{\R}\right]^{-\sigma_3} & \text{in } K^{(2)},\\
    \left[ -B^{\R}\right]^{\sigma_3} & \text{in } K^{(3)},
\end{dcases}
\end{aligned} \label{asymptotics of Pgl near 1}
\end{equation}
where $\Theta_1(\lambda)$ is an analytic function in a neighborhood of $\overline U_{1}$ such that \( \Theta_1(1) = 2^\nu \).

Similarly, we define the following sub-domains $K^{(4,5,6)}$ of a sufficiently small neighborhood $U_{-1}$ of $\lambda = -1$:
\begin{align*}
    K^{(4)} &= U_{-1} \cap \text{(outside of the unit circle)}\\
    K^{(5)} &= U_{-1} \cap \text{(inside of the unit circle)} \cap \text{(upper half plane)}\\
    K^{(6)} &= U_{-1} \cap \text{(inside of the unit circle)} \cap \text{(lower half plane)}.
\end{align*}
\begin{figure}[H]
    \centering
    \begin{subfigure}{.33\textwidth}
    \includegraphics[width=0.9\linewidth]{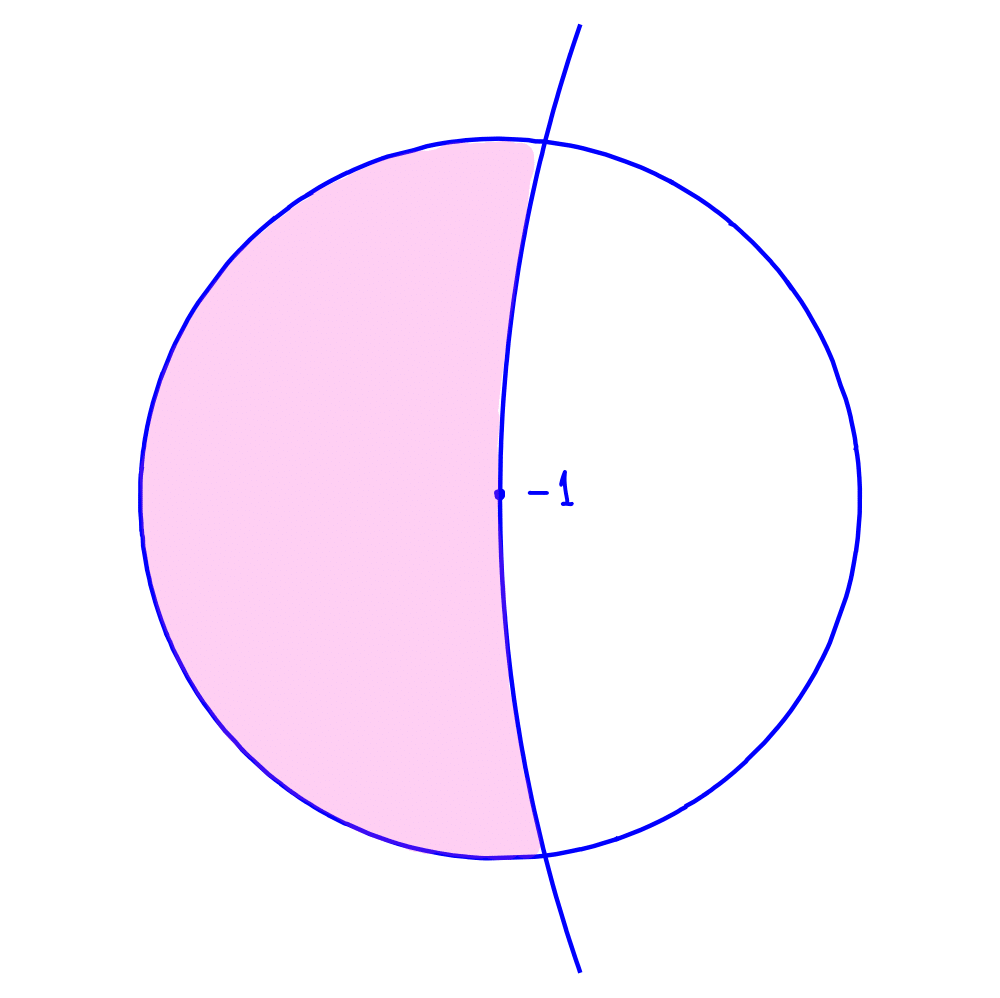}
    \caption{$K^{(4)}$}
    \label{Picture of K^{(4)}}
    \end{subfigure}%
    \begin{subfigure}{.33\textwidth}
    \includegraphics[width=0.9\linewidth]{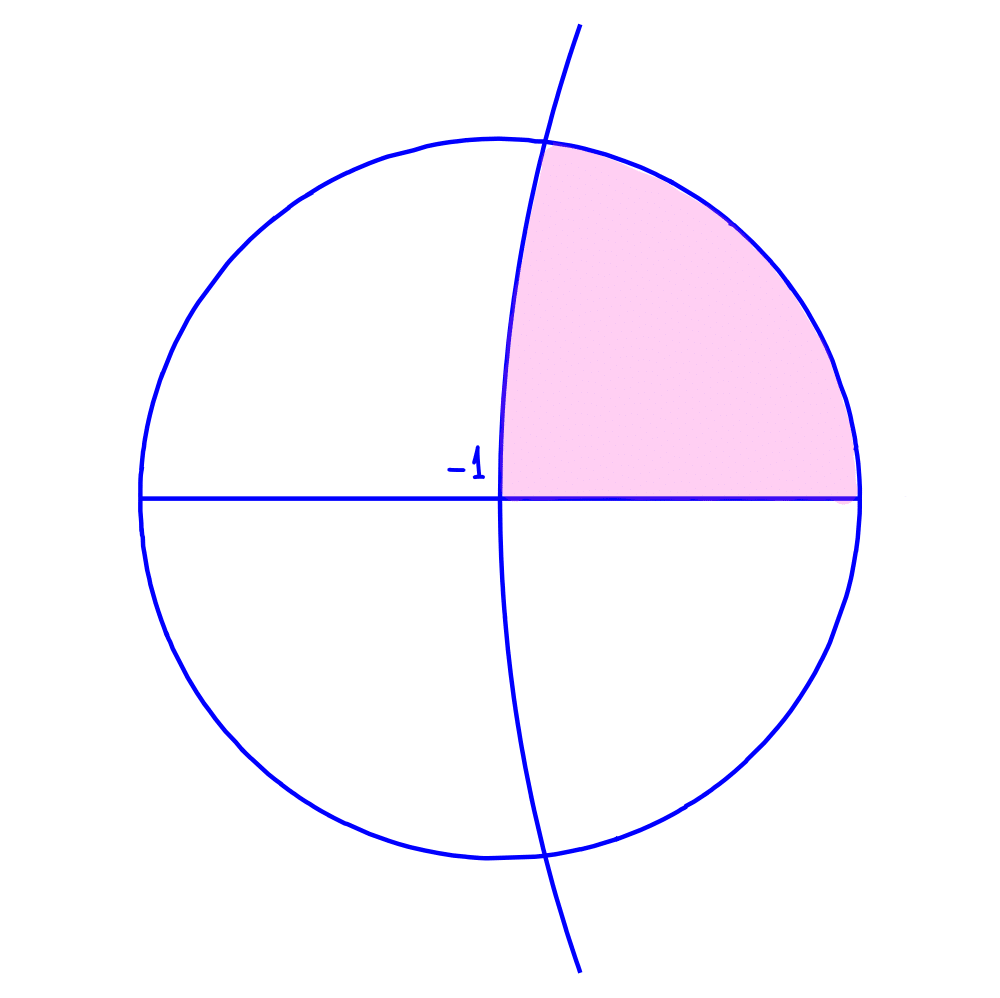}
    \caption{$K^{(5)}$}
    \label{Picture of K^{(5)}}
    \end{subfigure}
    \begin{subfigure}{.33\textwidth}
    \includegraphics[width=0.9\linewidth]{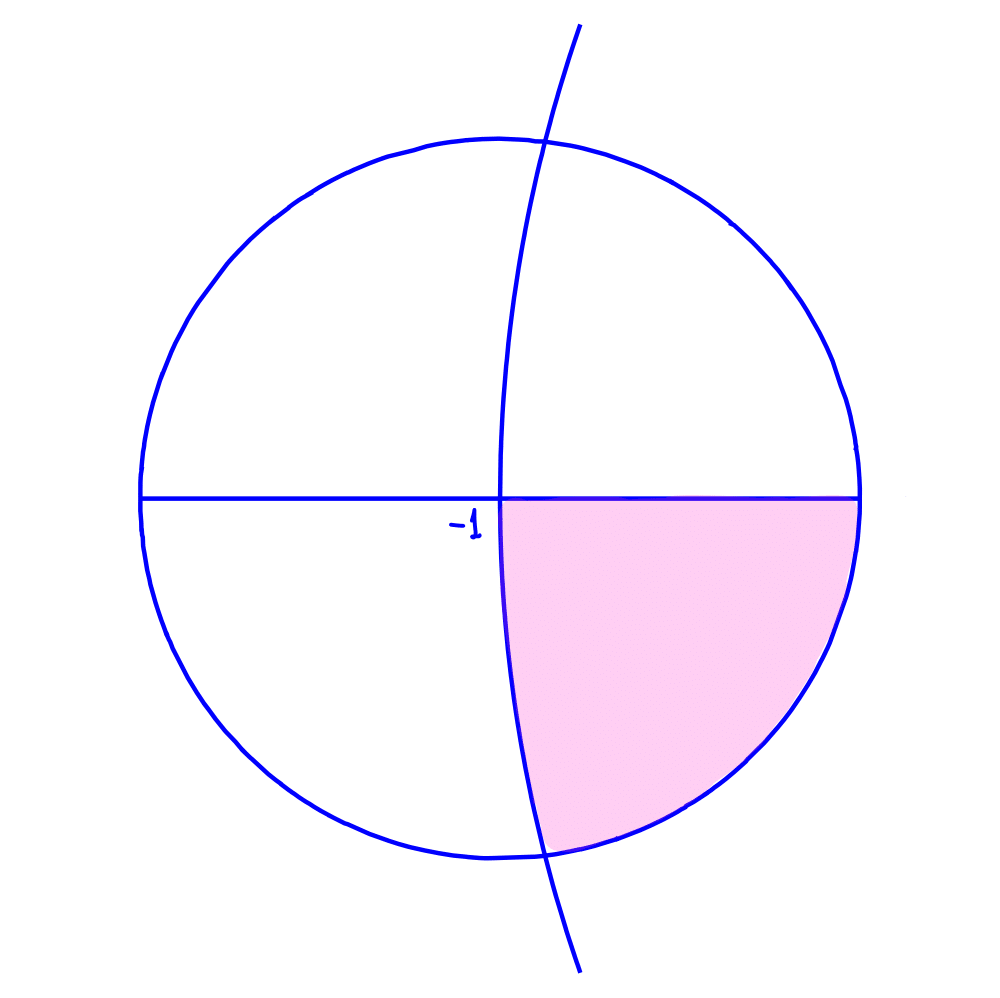}
    \caption{$K^{(6)}$}
    \label{Picture of K^{(6)}}
    \end{subfigure}
    \caption{Regions $K^{(4,5,6)}$.}
    \label{Pictures of K^{(ii)}}
\end{figure}
Let \( (\lambda+1)^\nu \) be the branch with the cut \( \{ \lambda \geq -1\} \) where we take \(\arg(\lambda+1) \in (0,2\pi) \). Then, one can write the global parametrix $P^{(gl)}(\lambda)$ in \( U_{-1} \) as
\begin{equation}
\begin{aligned}
P^{(gl)}(\lambda) = \Theta_{-1} (\lambda)^{\sigma_3} (\lambda + 1)^{-\nu \sigma_3} C^{(-1)}, \quad C^{(-1)} :=
\begin{dcases}
    I & \text{in } K^{(4)},\\
    \left[ B^{\R}\right]^{-\sigma_3} & \text{in } K^{(5)},\\
    \left[ -B^{\R}\right]^{\sigma_3} & \text{in } K^{(6)},
\end{dcases}
\end{aligned} \label{asymptotics of Pgl near -1}
\end{equation}
where $\Theta_{-1}(\lambda)$ is analytic in a neighborhood of $\overline U_{-1}$ and \( \Theta_{-1}(-1) = 2^\nu e^{\pi i\nu} \).

%%%%%%%%%%%%%%%%%%%%%%%%%%%%%%%

\subsubsection{Local Parametrix at \( \lambda = 1 \)} \label{local param}

Since the jump matrices forming \( G_{\tilde X} \) are not uniformly close to the identity matrix in \( U_1 \), we need to solve RHP~\ref{X tilde RH problem} locally in \( U_1 \). There are many such solutions as any of them can be multiplied by a holomorphic matrix function on the left. We seek one that matches well the global parametrix on \( \partial U_1 \).

To find the desired local parametrix, let us observe that the jump relations in RHP~\ref{X tilde RH problem}(2) can be changed within \( U_1 \) into the ones depicted on Figure~\ref{jump matrices for Phi near 1} 
\begin{figure}[ht!]
    \renewcommand\thesubfigure{\Alph{subfigure}}

    \centering
    \begin{subfigure}{.5\textwidth}
    \includegraphics[width=1\linewidth]{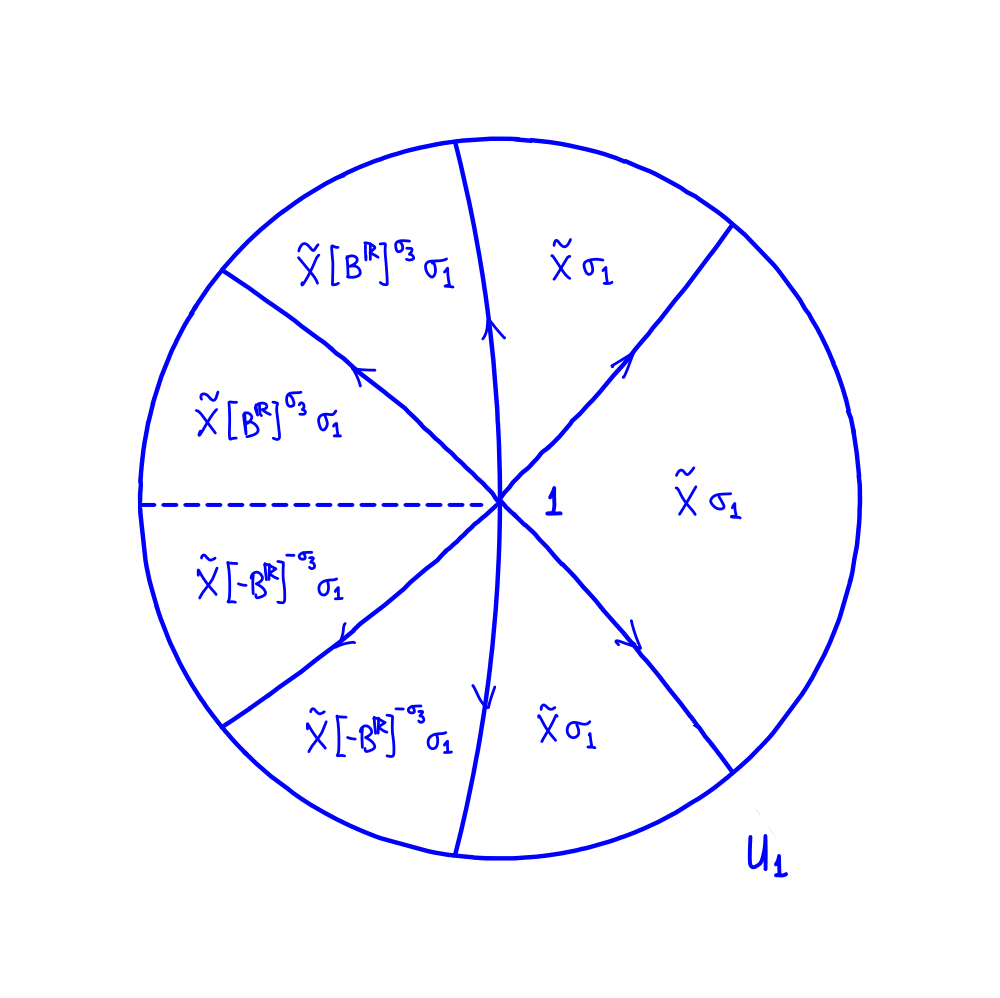}
    \caption{}
    \label{defn of Phi near 1}
    \end{subfigure}%
    \begin{subfigure}{.5\textwidth}
    \includegraphics[width=1\linewidth]{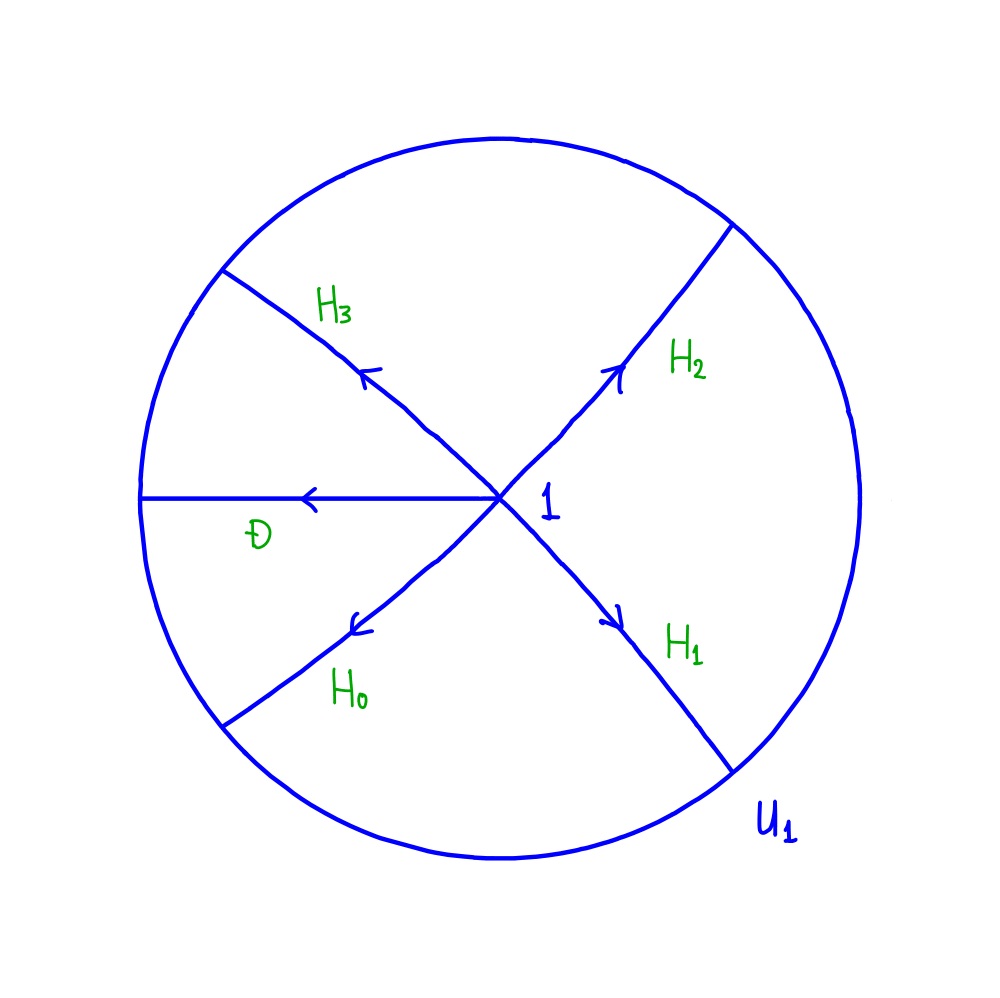}
    \caption{}
    \label{jump matrices for Phi near 1}
    \end{subfigure}
    \caption{Local transformation of \( \tilde X(\lambda) \mapsto \tilde X(\lambda)C^{(1)-1}\sigma_1 \) and the resulting jump relations.}
\end{figure}
via the transformation \( \tilde X(\lambda) \mapsto \tilde X(\lambda)C^{(1)-1}\sigma_1 \), see Figure~\ref{defn of Phi near 1}, where \( C^{(1)}\) was introduced in \eqref{asymptotics of Pgl near 1} and the resulting jump matrices are given by
\begin{align}
\begin{aligned}
&H_0 = \begin{pmatrix}
    1 & 0\\
    \frac{\overline{A}}{B^{\R}} e^{2 \pi i \nu} e^{2x \varphi} & 1
\end{pmatrix}, \quad
H_1 = \begin{pmatrix}
    1 & \frac{A}{B^{\R}} e^{-2x \varphi}\\
    0 & 1
\end{pmatrix}, \\
&H_2 = \begin{pmatrix}
    1 & 0\\
    -\frac{\overline{A}}{B^{\R}} e^{2x \varphi} & 1
\end{pmatrix}, \quad
H_3 = \begin{pmatrix}
    1 & -\frac{A}{B^{\R}} e^{2\pi i \nu} e^{-2x \varphi}\\
    0 & 1
\end{pmatrix},\\
& D = e^{2 \pi i \nu \sigma_3}, 
\end{aligned} \label{Phi jump matrices}
\end{align}
and $\nu$ was defined in \eqref{defn of nu}. It readily follows from \eqref{defn X} that
\begin{align}
    2 \varphi(\lambda) = i + \frac{i}{2}(\lambda - 1)^2 + \O\left((\lambda - 1)^3\right) \quad \text{as} \quad \lambda\to 1.
    \label{2phi near 1}
\end{align}
Hence, the quadratic local form of $2 \varphi(\lambda)$ and the structure of the jump matrices \eqref{Phi jump matrices} lead us to the following Riemann-Hilbert problem: 
\begin{RHP}\label{parabolic cylinder rhp}
Find a $2 \times 2$ matrix function $Z(z)$ such that
\begin{enumerate}
    \item $Z(z)$ is analytic for $z \in \C \setminus \Gamma_{Z}$, where $\Gamma_Z$ is depicted in Figure~\ref{pic for the Z problem};
    \item $Z(z)$ has continuous traces on $\Gamma_Z\setminus\{0\}$ that satisfy
    \begin{align*}
        Z_+(z) = Z_-(z) G_{Z}, \quad z \in \Gamma_{Z},
    \end{align*}
    where the jump matrices comprising $G_{Z}$ are assigned as on Figure \ref{pic for the Z problem} and given by
    \begin{align}
    \begin{aligned}
        &\widetilde H_0 = \begin{pmatrix}
        1 & 0\\
        - s_2 e^{2 \pi i \nu} e^{-\frac{z^2}{2}} & 1
        \end{pmatrix}, \quad
        \widetilde H_1 = \begin{pmatrix}
        1 & s_1 e^{\frac{z^2}{2}}\\
        0 & 1
        \end{pmatrix}, \\
        &\widetilde H_2 = \begin{pmatrix}
        1 & 0\\
        s_2 e^{-\frac{z^2}{2}} & 1
        \end{pmatrix}, \quad
        \widetilde H_3 = \begin{pmatrix}
        1 & - s_1 e^{2\pi i \nu} e^{\frac{z^2}{2}}\\
        0 & 1
        \end{pmatrix},\\
        & \widetilde D = e^{2 \pi i \nu \sigma_3}, \quad s_1 = \frac{A}{B^{\R}} e^{-ix}, \quad s_2 = -\frac{\overline{A}}{B^{\R}} e^{ix};
    \end{aligned}
    \end{align}
    \item it holds that
    \begin{align*}
        Z(z) = (I + \O(1/z)) z^{-\nu \sigma_3} \quad \text{as} \quad z\to\infty.
    \end{align*}
\end{enumerate}
\end{RHP}
\begin{figure}[htbp]
    \centering
    \includegraphics[width=6cm]{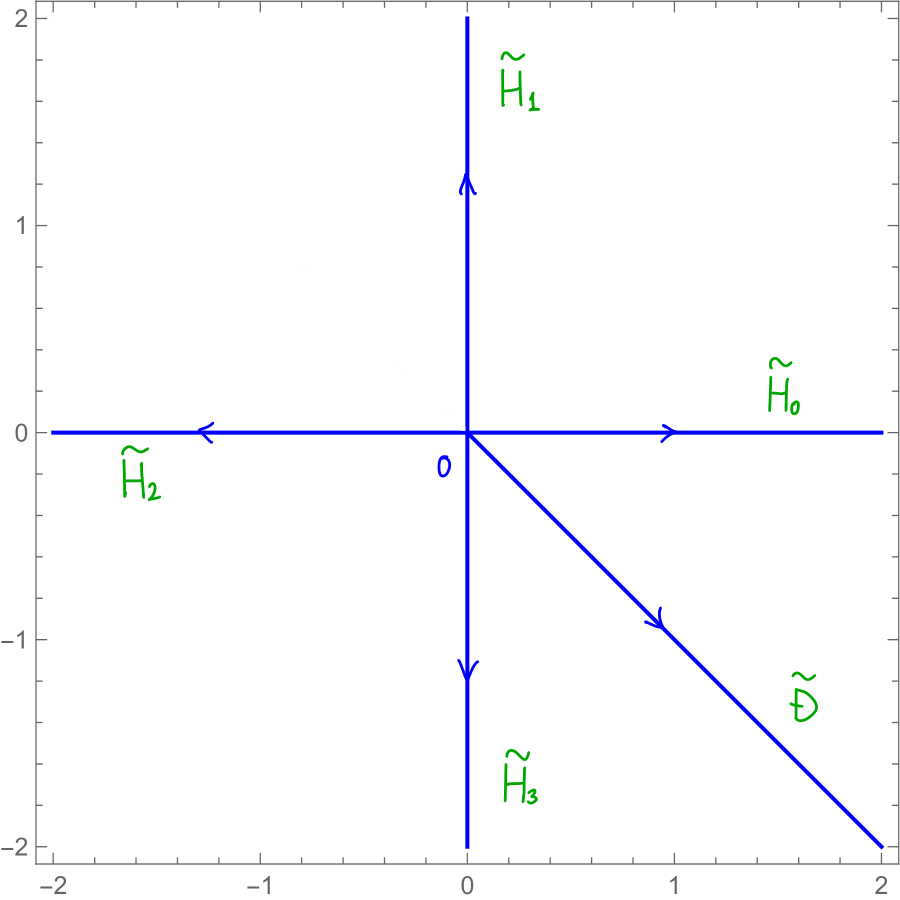}
    \caption{The contour $\Gamma_Z$, which consists of coordinate axes and the ray \( \arg(z) = -\pi /4 \), and the jump matrices $G_{Z}$ for RHP~\ref{parabolic cylinder rhp}.}
    \label{pic for the Z problem}
\end{figure}

RHP~\ref{parabolic cylinder rhp} admits an explicit solution in terms of the parabolic cylinder functions. Namely, if we set \( \arg(z)\in(-\frac\pi 4,\frac {7\pi}4 ) \), then \( Z(z)e^{\frac{z^2}4\sigma_3} \) is equal to
\[
 \begin{cases}
    \begin{pmatrix} D_{-\nu}\left(ze^{\frac{\pi i}2} \right )e^{\frac{\pi i\nu }2} & -\alpha D_{\nu-1}(z) \\ i\beta D_{-\nu-1}\left(ze^{\frac{\pi i}2}\right )e^{\frac{\pi i\nu }2} & D_{\nu}(z) \end{pmatrix}, & \arg(z)\in\left(-\frac\pi4,0\right), \medskip \\
    \begin{pmatrix} D_{-\nu}\left(ze^{-\frac{\pi i}2} \right )e^{-\frac{\pi i\nu }2} & -\alpha D_{\nu-1}(z) \\ -i\beta D_{-\nu-1}\left(ze^{-\frac{\pi i}2}\right )e^{-\frac{\pi i\nu }2} & D_{\nu}(z) \end{pmatrix}, & \arg(z)\in\left(0,\frac\pi2\right), \medskip \\
    \begin{pmatrix} D_{-\nu}\left(ze^{-\frac{\pi i}2} \right )e^{-\frac{\pi i\nu }2} & \alpha D_{\nu-1}\left(ze^{-\pi i} \right)e^{\pi i \nu} \\ -i\beta D_{-\nu-1}\left(ze^{-\frac{\pi i}2}\right )e^{-\frac{\pi i\nu }2} & D_{\nu}\left(ze^{-\pi i} \right)e^{\pi i \nu} \end{pmatrix}, & \arg(z)\in\left(\frac\pi2,\pi\right), \medskip \\
    \begin{pmatrix} D_{-\nu}\left(ze^{-\frac{3\pi i}2} \right )e^{-\frac{3\pi i\nu }2} & \alpha D_{\nu-1}\left(ze^{-\pi i} \right)e^{\pi i \nu} \\ i\beta D_{-\nu-1}\left(ze^{-\frac{3\pi i}2}\right )e^{-\frac{3\pi i\nu }2} & D_{\nu}\left(ze^{-\pi i} \right)e^{\pi i \nu} \end{pmatrix}, & \arg(z)\in\left(\pi,\frac{3\pi}2\right), \medskip \\
    \begin{pmatrix} D_{-\nu}\left(ze^{-\frac{3\pi i}2} \right )e^{-\frac{3\pi i\nu }2} & -\alpha D_{\nu-1}\left(ze^{-2\pi i} \right)e^{2\pi i \nu} \\ i\beta D_{-\nu-1}\left(ze^{-\frac{3\pi i}2}\right )e^{-\frac{3\pi i\nu }2} & D_{\nu}\left(ze^{-2\pi i} \right)e^{2\pi i \nu} \end{pmatrix}, & \arg(z)\in\left(\frac{3\pi}2,\frac{7\pi}{4}\right),
\end{cases}
\]
where
\[
\alpha = -e^{-2\pi i\nu}\frac{i}{s_2}\frac{\sqrt{2\pi}}{\Gamma(\nu)} = e^{-\pi i \nu} s_1 \frac{\Gamma(1-\nu)}{\sqrt{2\pi}} \quad \text{and} \quad \alpha\beta = - \nu.
\]
Recall that \(\overline A = p B^\R \), \( |A|^2 = 1 + (B^\R)^2 \), \(\kappa = -\frac1\pi \ln B^\R\), and \( \nu = \frac12+i\kappa \). Therefore,
\begin{align}
\alpha &= e^{\frac{\pi i}2 - ix - 2\pi i\nu }\frac{B^\R}{\overline A}  \frac{\sqrt{2\pi}}{\Gamma(\nu)} \\
&= e^{- ix + 2\pi \kappa -\frac{\pi i}2 } \frac{B^\R}{\overline A} \frac{|A|}{\sqrt{B^\R}} e^{-i\arg(\Gamma(\nu))} \\
& = \big( B^\R\big)^{-3/2} e^{-i(x + \arg(p\Gamma(\nu)) + \frac\pi 2)},
\label{defn of alpha}
\end{align}
where we used the identity \( \pi |\Gamma(\frac12+ i \kappa)|^{-2} = \cosh(\pi \kappa)\) for the first equality. 

The fact that the above-defined matrix functions satisfy RHP~\ref{parabolic cylinder rhp}(2) can be easily verified using the identity
\[
D_\mu(z) = \frac{\Gamma(\mu+1)}{\sqrt{2\pi}} \left( e^{\frac{\pi i \mu}2} D_{-\mu-1}\left( ze^{\frac{\pi i}2}\right) + e^{-\frac{\pi i \mu}2} D_{-\mu-1}\left( ze^{-\frac{\pi i}2}\right) \right).
\]
Furthermore, using the asymptotic expansion formula
\[
D_\mu(z) \sim z^\mu \left( 1 - \frac{\mu(\mu-1)}{2z^2} + \frac{\mu(\mu-1)(\mu-2)(\mu-3)}{8z^4} + \cdots \right) e^{-\frac{z^2}4},
\]
which holds uniformly in sectors \( |\arg(z)|\leq \frac{3\pi}4 - \delta\), \( \delta>0 \), one can check that
\begin{equation}
\label{matrix Z}
Z(z) = \left( I + \frac1z \begin{pmatrix} 0 & -\alpha \\ \beta & 0 \end{pmatrix} + \frac1{z^2} \begin{pmatrix} \frac{\nu(1+\nu)}2 & 0 \\ 0 & \frac{\nu(1-\nu)}2 \end{pmatrix} + \O\left(\frac1{z^3}\right) \right) z^{-\nu\sigma_3}
\end{equation}
uniformly for \( z\to\infty \), which is a refinement of RHP~\ref{parabolic cylinder rhp}(3).

To set up correspondence between \( \lambda \) and \( z \) planes, we introduce the following conformal mapping on $U_1$:
\begin{equation}
\label{Z <--> lambda}
    z_1(\lambda) := e^{\frac{\pi i}{2}} \sqrt{4 \varphi(\lambda) - 2i} = e^{\frac{3\pi i}{4}} \left(\lambda - 1 + \O\big((\lambda-1)^2\big) \right) \quad \text{as} \quad \lambda\to1,
\end{equation}
where asymptotics at \( 1 \) follows from \eqref{2phi near 1} and also fixes the branch of the square root used in \eqref{Z <--> lambda}. Recall that the contour \( \Gamma_{\tilde X}\cap U_1 \) consists of the curves on which the real or imaginary parts of \( 2\varphi(\lambda) - i \) vanish. Hence, one can easily see that \( x^{1/2} z_1(\lambda )\) maps the contour on Figure~\ref{jump matrices for Phi near 1} into \( \Gamma_Z \).

We now define local parametrix \( P^{(1)}(\lambda) \) by setting
\[
P^{(1)}(\lambda) := V^{(1)}(\lambda) Z\left(x^{1/2} z_1(\lambda)\right) \sigma_1 C^{(1)}, \quad \lambda \in U_1,
\]
where the matrix \( Z(x^{1/2} z_1(\lambda)) \sigma_1 C^{(1)} \) has the same jumps in \( U_1 \) as \( \tilde X(\lambda) \) and \( V^{(1)}(\lambda) \) is a holomorphic prefactor, see \eqref{asymptotics of Pgl near 1}, given by
\[
    V^{(1)}(\lambda) := P^{(gl)}(\lambda) \, \left(x^{1/2} z_1(\lambda)\right)^{- \nu \sigma_3} C^{(1)-1} \sigma_1.
\]

Since $B^{\R} > 0$ and therefore $\Re \nu = \frac{1}{2}$, direct computations using \eqref{asymptotics of Pgl near 1} and \eqref{matrix Z} yield the following lemma.

\begin{lem} \label{lemma 2.3} 
It holds uniformly for \( \lambda\in\partial U_1 \) that
\[
    P^{(gl)}(\lambda) \left[ P^{(1)}(\lambda) \right]^{-1} = I + \alpha \, x^{\nu-\frac12} \, \frac{p_1(\lambda)}{\lambda-1} \begin{pmatrix}
        0 & 0 \\
        1 & 0
    \end{pmatrix} + \O\left( \frac1x \right)
\]
as \( x\to \infty \), where \( p_1(\lambda) \) is a holomorphic function on \( \overline U_1 \) given by
\[
p_1(\lambda) = \Theta_1(\lambda)^2 \left( \frac{z_1(\lambda)}{\lambda-1} \right)^{2\nu-1}.
\]
In particular, it follows from \eqref{asymptotics of Pgl near 1} and \eqref{Z <--> lambda} that
\begin{equation}
\label{p1(1)}
p_1(1) = 2^{2\nu} e^{\frac{3\pi i}{4}(2\nu-1)} = 2\big( B^\R \big )^{3/2} e^{i\kappa \ln 4}.    
\end{equation}
\end{lem}

\subsubsection{Local Parametrix at \( \lambda = -1 \)} \label{local param 2}

The local parametrix $P^{(-1)}(\lambda)$ is constructed similarly. First, we transform \( \lambda \)  plane into \( z \) plane via the conformal map
\begin{align}
\label{Z <--> lambda -1}
    z_{-1}(\lambda) := e^{-\pi i} \sqrt{4\varphi(\lambda) + 2i} = e^{-\frac{\pi i}{4}} (\lambda + 1) + \O\big((\lambda+1)^2\big),
\end{align}
where the asymptotic formula holds as \( \lambda\to -1\) and the value \( z_{-1}^\prime(-1)=e^{-\frac{ \pi i}{4}} \) fixes the branch of the square root used to define the map. Again, one can see that \( x^{1/2} z_{-1}(\lambda) \) takes the stationary contour \( \Im(2\varphi+i) = 0\) into the coordinate axes.

Now, using the solution \( Z(z) \) of RHP~\ref{parabolic cylinder rhp}, the matrix \( C^{(-1)}\) from \eqref{asymptotics of Pgl near -1}, and the above conformal map, we define
\[
P^{(-1)}(\lambda) := V^{(-1)}(\lambda)Z\left(x^{1/2}z_{-1}(\lambda)\right)C^{(-1)}, \quad \lambda\in U_{-1},
\]
where the matrix \( Z(x^{1/2}z_{-1}(\lambda))C^{(-1)} \) has the same jumps on \( \Gamma_{\tilde X}\cap U_{-1} \) as \( \tilde X(\lambda)\) and \( V^{(-1)}(\lambda) \) is a holomorphic prefactor given by
\begin{align}
    V^{(-1)}(\lambda) = P^{(gl)}(\lambda) \, \left( x^{1/2} z_{-1}(\lambda)\right)^{\nu \sigma_3}C^{(-1)-1}.
    \label{defn of V check}
\end{align}
Similarly to Lemma~\ref{lemma 2.3}, we then conclude from \eqref{asymptotics of Pgl near -1} and \eqref{matrix Z} that the following asymptotics holds.

\begin{lem}\label{lemma 2.4} 
It holds uniformly for \( \lambda\in\partial U_{-1} \) that
\begin{align}
    P^{(gl)}(\lambda) \left[ P^{(-1)}(\lambda) \right]^{-1} = I + \alpha \, x^{\nu-\frac12}\,\frac{p_{-1}(\lambda)}{\lambda+1}\begin{pmatrix}
        0 & 1 \\
        0 & 0
    \end{pmatrix} + \O\left( \frac{1}{x} \right)
\end{align}
as \( x\to\infty \), where \( p_{-1}(\lambda) \) is a holomorphic function on \( \overline U_{-1} \) given by
\[
p_{-1}(\lambda) = \Theta_{-1}(\lambda)^2\left( \frac{z_{-1}(\lambda)}{\lambda+1} \right)^{2\nu-1}.
\]
In particular, it follows from \eqref{asymptotics of Pgl near -1} and \eqref{Z <--> lambda -1} that
\begin{align}
    p_{-1}(-1) = \left(2 \, e^{\pi i}\right)^{2 \nu} e^{-\frac{\pi i}{4}(2\nu-1)} = -p_1(1).
\end{align}
\end{lem}

%%%%%%%%%%%%%%%%%%%%%%%%%%%%%%%

\subsubsection{Small Norm Problem} \label{dressing}

We look for the solution of RHP~\ref{X tilde RH problem} in the form
\begin{equation}
\label{form of tilde X}
    \tilde{X} (\lambda) = R(\lambda)\begin{dcases}
        P^{(gl)}(\lambda) & \textit{if } \lambda \in \C \setminus (S^1 \cup U_1 \cup U_{-1}),\\
        P^{(\pm1)}(\lambda) & \textit{if } \lambda \in U_{\pm1},
    \end{dcases}
\end{equation}
where the error function \( R(\lambda)\) solves of the following Riemann-Hilbert problem.
\begin{RHP}\label{R-rhp}
Find a $2 \times 2$ matrix function $R(\lambda)$ such that
\begin{enumerate}
    \item $R(\lambda)$ is analytic for $\lambda \in \C \setminus \Gamma_R$, where \( \Gamma_R \), see Figure \ref{pic for the R problem}, is given by 
    \[ 
    \Gamma_R:= \Gamma_R^* \cup \partial U_1 \cup \partial U_{-1}, \quad \Gamma_R^* := \Gamma_{\tilde X}\setminus (S^1\cup \overline U_1 \cup \overline U_{-1}); 
    \]
    \item one-sides traces $R_\pm(\lambda)$ exists a.e. on \( \Gamma_R \), belong to \( L^2(\Gamma_R) \), and satisfy
    \begin{align*}
        R_+(\lambda) = R_-(\lambda) G_{R}, \quad \lambda \in \Gamma_R,
    \end{align*}
    where
    \begin{align}
        G_R := \begin{dcases}
            P^{(gl)}(\lambda) G_{\tilde{X}} \left[ P^{(gl)}(\lambda) \right]^{-1} & \text{if } \lambda \in \Gamma_R^*,\\
            P^{(gl)}(\lambda) \left[ P^{(\pm1)}(\lambda) \right]^{-1} & \text{if } \lambda \in \partial U_{\pm1};
        \end{dcases}
    \end{align}
    \item it holds that
    \begin{align*}
        R(\lambda) = I + \O(1/\lambda) \quad \text{as} \quad \lambda\to\infty.
    \end{align*}
\end{enumerate}
\end{RHP}
\begin{figure}[htbp]
    \centering
    \includegraphics[width=6cm]{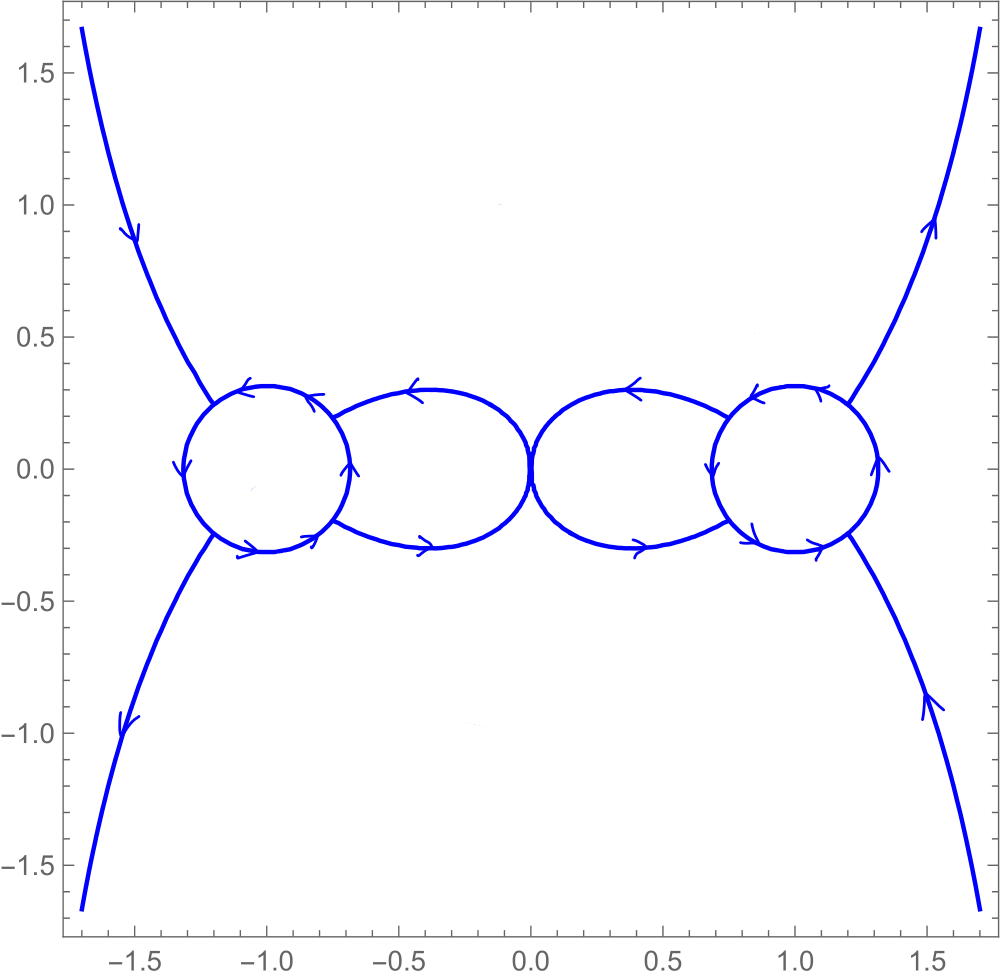}
    \caption{The contour $\Gamma_R$ for RHP~\ref{R-rhp}.}
    \label{pic for the R problem}
\end{figure}

It readily follows from \eqref{defn of alpha} and Lemmas~\ref{lemma 2.3}--\ref{lemma 2.4} that the jump \( G_R \) is not uniformly close to \( I \) on \(\partial U_1 \cup \partial U_{-1} \) as \( |x^{\nu-\frac12}|=1 \). To circumvent this difficulty, let
\[
E^{(1)}(\lambda) := I - \alpha x^{\nu-\frac12}\frac{p_1(\lambda)}{\lambda-1} \begin{pmatrix}
        0 & 0 \\
        1 & 0
    \end{pmatrix}, \quad 
E^{(-1)}(\lambda) := I - \alpha x^{\nu-\frac12}\frac{p_{-1}(\lambda)}{\lambda+1}\begin{pmatrix}
        0 & 1 \\
        0 & 0
    \end{pmatrix}.
\]
We look for the solution of RHP~\ref{R-rhp} in the form
\begin{align}
    R(\lambda) = \Tilde R(\lambda) \begin{dcases}
        E^{(\pm1)}(\lambda)^{-1} & \text{if } \lambda \in U_{\pm1}, \\
        I & \text{otherwise},
    \end{dcases}
    \label{defn of R tilde}
\end{align}
where $\tilde{R}(\lambda)$ solves the following Riemann-Hilbert problem.
\begin{RHP}\label{R tilde rhp}
Find a $2 \times 2$ matrix function $\Tilde{R}(\lambda)$ such that
\begin{enumerate}
    \item $\Tilde{R}(\lambda)$ is analytic for $\lambda \in \C \setminus (\Gamma_R \cup \{\pm 1\})$, where \( \Gamma_R\) is the same as in RHP~\ref{R-rhp};
    \item one-sides traces $\Tilde R_\pm(\lambda)$ exists a.e. on \( \Gamma_R \), belong to \( L^2(\Gamma_R) \), and satisfy
    \begin{align*}
        \Tilde{R}_+(\lambda) = \Tilde{R}_-(\lambda) G_{\Tilde{R}}, \quad \lambda \in \Gamma_R,
    \end{align*}
    where
    \begin{align}
        G_{\Tilde{R}} = G_R \begin{dcases}
             I & \text{ if } \lambda\in \Gamma_R^*, \\
             E^{(\pm1)}(\lambda) & \text{if } \lambda \in \partial U_{\pm1};
        \end{dcases} \label{jump matrices G R tilde}
    \end{align}
    \item it holds that
    \begin{align*}
        \Tilde{R}(\lambda) = I + \O(1/\lambda) \quad \text{as} \quad \lambda\to\infty;
    \end{align*}
        \item $\Tilde{R}(\lambda) = (\Tilde{R}^{(1)}(\lambda), \Tilde{R}^{(2)}(\lambda))$ has simple poles at $\lambda = \pm 1$ with residues
    \begin{align}
        \res_{\lambda = 1} \Tilde{R}^{(1)}(\lambda) &= - \Tilde{\alpha} \Tilde{R}^{(2)}(1) \label{residue relation 1},\\
        \res_{\lambda = -1} \Tilde{R}^{(2)}(\lambda) &=  \Tilde{\alpha} \Tilde{R}^{(1)}(-1), \label{residue relation 2}
    \end{align}
    where, see \eqref{kappa phi}, \eqref{defn of alpha}, and \eqref{p1(1)}, we have that
    \begin{equation}
    \label{tilde alpha}
        \Tilde\alpha := \alpha x^{\nu-\frac12}p_1(1) = 2e^{-i( x - \kappa \ln x + \phi+\frac\pi2)}.
    \end{equation} 
\end{enumerate}
\end{RHP}

It readily follows from Lemmas~\ref{lemma 2.3}--\ref{lemma 2.4} and \eqref{tilde alpha} that 
\begin{equation}
\label{G jump 1}
\|G_{\Tilde R} - I \|_{L^\infty(\partial U_1\cup \partial U_{-1})} = \O\left( x^{-1} \right) \quad \text{as} \quad x\to\infty.    
\end{equation}
Moreover, it follows from \eqref{L1, D1, R1}, \eqref{L2, D2, R2}, \eqref{L3, D3, R3}, \eqref{L4, D4, R4} that the absolute value of the only non-zero entry of \( G_{\Tilde R}(\lambda) - I \)  on \( \Gamma_R^* \) is equal to \( |p|e^{-x|\Re(2\varphi)|} \), see Figure~\ref{X tilde problem pic}. Recall also that \( \Gamma_R^* \) is the set where \( |\Im(2\varphi)| = 1 \). We then get from \eqref{real part varphi} that
\begin{equation}
\label{real part varphi 2}
|\Re(2\varphi(\lambda))| = |\eta| \frac{|1-|\lambda|^2|}{2|\lambda|^2} = \frac1{|\xi|^{1/2}} \frac{|1-|\xi||}{\sqrt{2-|\xi|}}\frac{|1-|\lambda|^2|}{1+|\lambda|^2},
\end{equation}
where \( \lambda = \xi+i\eta \). This expression shows that \( |\Re(2\varphi(\lambda))|\geq c>0 \) on \( \Gamma_R^* \), where  constant \( c \) depends on the radii of \( U_1 \) and \( U_{-1} \). Observe also that \( \Gamma_R \) lies within the strip \(|\xi|<2\) and is asymptotic to the vertical lines \(\xi=\pm2\). Hence, we can conclude that
\begin{equation}
\label{G jump 2}
\|G_{\Tilde R} - I\|_{L^q(\Gamma_R^*)} = \O\left( e^{-cx} \right) \quad \text{as} \quad x\to\infty
\end{equation}
for a possibly adjusted constant \( c \) and any \( q\in[1,\infty]\). Altogether, we have that RHP~\ref{R tilde rhp} is a small norm problem. However, \( \Tilde R(\lambda) \) has poles. Following the dressing technique of \cite{BothnerIts}, we look for the solution of RHP~\ref{R tilde rhp} in the form
\begin{align}
    \Tilde{R}(\lambda) = (\lambda I + B) \hat{R}(\lambda) \begin{pmatrix}
        \lambda - 1 & 0 \\
        0 & \lambda + 1
    \end{pmatrix}^{-1},
    \label{defn of R check}
\end{align}
where the matrix $B$ will be specified shortly and $\hat{R}(\lambda)$ solves the following (small norm) Riemann-Hilbert problem.
\begin{RHP}\label{R hat rhp}
Find a $2 \times 2$ matrix function $\hat{R}(\lambda)$ such that
\begin{enumerate}
    \item $\hat{R}(\lambda)$ is analytic for $\lambda \in \C \setminus \Gamma_R$, where $\Gamma_R$ is the same as in RHP~\ref{R-rhp};
    \item one-sides traces $\hat R_\pm(\lambda)$ exists a.e. on \( \Gamma_R \), belong to \( L^2(\Gamma_R) \), and satisfy
    \begin{align*}
        \hat{R}_+(\lambda) = \hat{R}_-(\lambda) G_{\hat{R}}, \quad \lambda \in \Gamma_R,
    \end{align*}
    where
    \begin{align}
        G_{\hat{R}} = \begin{pmatrix}
        \lambda - 1 & 0 \\
        0 & \lambda + 1
    \end{pmatrix}^{-1} G_{\tilde{R}}
    \begin{pmatrix}
        \lambda - 1 & 0 \\
        0 & \lambda + 1
    \end{pmatrix}; \label{jump matrices G R hat}
    \end{align}
    \item it holds that
    \begin{align*}
        \hat{R}(\lambda) = I + \O(1/\lambda) \quad \text{as} \quad \lambda\to\infty.
    \end{align*}
\end{enumerate}
\end{RHP}

This Riemann-Hilbert problem can be solved as in Section~\ref{case 1} because
\begin{equation} 
\label{small norm jump}
\| G_{\hat R} - I\|_{L^2(\Gamma_R)\cap L^\infty(\Gamma_R)} = \O\left( x^{-1} \right) \quad \text{as} \quad x\to\infty
\end{equation}
by \eqref{G jump 1} and \eqref{G jump 2}. Hence, provided we can choose matrix \( B \) so that \( \Tilde R(\lambda) \) as in \eqref{defn of R check} solves RHP~\ref{R tilde rhp}, we can solve all the previous Riemann-Hilbert problems. Matrix \( B \) is algebraically determined by residue conditions \eqref{residue relation 1} and \eqref{residue relation 2}. Write $\ds \hat{R}(\lambda) = (\hat{R}^{(1)}(\lambda), \hat{R}^{(2)}(\lambda))$. Then, 
\begin{align*}
    \left(\Tilde{R}^{(1)}(\lambda), \Tilde{R}^{(2)}(\lambda)\right) = (\lambda I + B) \left( \frac{1}{\lambda - 1}\hat{R}^{(1)}(\lambda), \frac{1}{\lambda + 1} \hat{R}^{(2)}(\lambda)\right)
\end{align*}
by \eqref{defn of R check}. It is a tedious but straightforward computation to verify that
\[
G_{\hat R}(\lambda) = \sigma_1 G_{\hat R}(-\lambda) \sigma_1.
\]
As a solution of a small norm Riemann-Hilbert problem, \( \hat R(\lambda) \) is unique and therefore \( \hat R(\lambda) = \sigma_1 \hat R(-\lambda) \sigma_1 \). Write \( \hat R(\lambda) = I + E(\lambda) \). Then,
\begin{equation}
    \label{symmerries of E}
    \hat E(\lambda) = \sigma_1 \hat E(-\lambda) \sigma_1.
\end{equation}
Write \( E(\lambda) = (E^{(1)}(\lambda),E^{(2)}(\lambda))\). Then, we get from \eqref{residue relation 1} and \eqref{residue relation 2} that
\begin{align}
\label{solve for B}
\begin{dcases}
    (I + B) \left( \mathbf{e}_1 + E^{(1)}(1) \right) = - \frac{\Tilde{\alpha}}{2} (I + B) \left( \mathbf{e}_2 + E^{(2)}(1) \right), \\
    (-I + B) \left( \mathbf{e}_2 + E^{(2)}(-1) \right) = -\frac{\Tilde{\alpha}}{2} (-I + B) \left( \mathbf{e}_1 + E^{(1)}(-1) \right),
\end{dcases}
\end{align}
where \(\mathbf{e}_1\) and \( \mathbf{e}_2 \) are the standard coordinate vectors in \( \R^2 \). Thus, symmetry \eqref{symmerries of E} at \( \lambda=1\) now implies that \( \)
\[
b_{11} = - b_{22} \quad \text{and} \quad b_{12} = - b_{21},  \quad B = \begin{pmatrix}
        b_{11} & b_{12}\\
        b_{21} & b_{22}
    \end{pmatrix}.    
\]
Since \( \hat R(\lambda) \) is a solution of a small norm problem, \( E(1) = \O(x^{-1}) \) as \( x\to\infty \) by \eqref{small norm jump}, see the next subsection. Because \( |\tilde \alpha|=2 \), see \eqref{tilde alpha}, we can further deduce that
\begin{align*}
    b_{11} = - b_{22} = - \frac{4 + \tilde{\alpha}^2 + \O(x^{-1})}{4 - \tilde{\alpha}^2 + \O(x^{-1})}, \quad b_{12} = -b_{21} = \frac{ 4\tilde{\alpha} + \O(x^{-1})}{4 - \tilde{\alpha}^2 + \O(x^{-1})}
\end{align*}
provided the denominator above is non-zero. Using \eqref{tilde alpha}, we then get that 
\begin{equation}
\label{bees}
ib_{11} = \frac{\sin(x-\kappa \ln x + \phi) + \O(x^{-1})}{\cos(x-\kappa \ln x+\phi) + \O(x^{-1})}, \quad ib_{12} = \frac{ 1 + \O(x^{-1})}{\cos(x-\kappa \ln x+\phi) + \O(x^{-1})}.
\end{equation}

%%%%%%%%%%%%%%%%%%%%%%%%%%%%%%%%

\subsubsection{Asymptotic Analysis}

Exactly as in Section~\ref{case 1}, by the small norm theorem, the unique solution $\hat{R}(\lambda)$ of  RHP~\ref{R hat rhp} is given by
\begin{align}
    \hat{R}(\lambda) &= I + \frac{1}{2 \pi i } \int_{\Gamma_R} \frac{\rho(\lambda') (G_{\hat{R}} (\lambda') - I)}{\lambda' - \lambda} d\lambda' \\
    & =  I + \frac{1}{2 \pi i } \int_{\Gamma_R} \frac{G_{\hat{R}} (\lambda') - I}{\lambda' - \lambda} d\lambda' + \frac{1}{2 \pi i } \int_{\Gamma_R} \frac{(\rho(\lambda') -I ) (G_{\hat{R}} (\lambda') - I)}{\lambda' - \lambda} d\lambda'
     \label{sing singular integral equation}
\end{align}
for \( \lambda \notin \Gamma_R \), where $\displaystyle{\rho(\lambda') := \hat{R}_-(\lambda^\prime)}$, $\lambda' \in \Gamma_R$, which itself is a solution of the corresponding integral equation. By \eqref{small norm jump}, the small norm theorem implies that
\begin{align}
    \|\rho - I\|_{L^2(\Gamma_R)} = \O \left( x^{-1} \right) \quad \text{as} \quad x\to\infty. \label{sing rho - I has a small norm}
\end{align}
In particular, we get that \( |E(1)| = \O\left( x^{-1} \right) \) by \eqref{small norm jump}, \eqref{sing rho - I has a small norm}, and the Cauchy-Schwarz inequality as claimed before. This means that if \( \{x_n\} \) is the increasing sequence of values of \( x \) for which the matrix \( B \) cannot be constructed, i.e., the denominator in \eqref{bees} vanishes, then
\[
x_n = \tilde x_n + \O\big(\tilde x_n^{-1}\big), \quad \tilde x_n - \kappa \ln \tilde x_n + \phi = \pi\left( n-\frac12 \right)
\]
for all \( n \) large enough. The above relations readily yield asymptotic formula \eqref{asymptotic of x_n}.

Next, we can infer from \eqref{real part varphi 2} that \( \Gamma_R \) is tangential to the imaginary axis at the origin and that \( G_{\hat R} - I \) vanishes exponentially there with respect to \( \eta \). Hence, as in Section~\ref{case 1}, \( \hat R(0) \) is well defined and it holds that
\[
   \lim_{\lambda\to0} \hat{R}(0) = I + \O (x^{-1})
\]
by \eqref{sing singular integral equation}, where the second integral is estimated by Cauchy-Schwarz inequality, \eqref{small norm jump}, and \eqref{sing rho - I has a small norm}, while \( L^q(\Gamma_R)\)-norms of \( (G_{\hat R}(\lambda^\prime)-I)/\lambda^\prime\) can be estimated around the origin using \eqref{G jump 2} and exponential vanishing of the integrand at \( \lambda=0 \). Recall now that \( P^{(gl)}(0) = i^{\sigma_3} \), see Lemma~\ref{lemma global par}. Then, it follows from \eqref{X tilde near zero}, \eqref{form of tilde X}, \eqref{defn of R tilde}, and \eqref{defn of R check} that
\begin{align*}
    P_0 &= \lim_{\lambda \to 0} \Tilde X (\lambda)\sigma_3\sigma_1 = \lim_{\lambda \to 0} R (\lambda) P^{(gl)}(\lambda) \sigma_3\sigma_1 \\ & = \lim_{\lambda \to 0} (\lambda I + B) \hat{R}(\lambda) \begin{pmatrix}
        \frac{1}{\lambda - 1} & 0 \\
        0 & \frac{1}{\lambda + 1}
    \end{pmatrix} i^{\sigma_3} \sigma_3 \sigma_1 \\ &= -B i^{\sigma_3} \sigma_1 + \O(x^{-1}).    
\end{align*}
Thus, we get from \eqref{defn of P0} and \eqref{bees} that
\begin{align}
    e^{\frac{u}{2}} & = \sinh \left(\frac u2\right) + \cosh \left(\frac u2\right) = ib_{12}-ib_{11} \\
    & = \frac{1 - \sin(x-\kappa \ln x + \phi) + \O(x^{-1})}{\cos(x-\kappa \ln x+\phi) + \O(x^{-1})}.
    \label{eu2}
\end{align}
Since we defined real solutions modulo addition of integer multiples of \( 2\pi i\), the above formula readily yields \eqref{asymp for B neq 0}.

\appendix

\section{Connection Problem for Real Solutions of Painlev\'e III}
\label{connection P3}

In this appendix we connect the behavior as \( x\to+\infty\) and \( x\to 0^+\) of the real solutions \( w(x) \) of \eqref{PIIID6} with $\alpha = \beta = 0$ and $\gamma = -\delta = 1/4$. Recall that if \( w(x)\) is such a solution, then so are \( -w(x) \) and \( \pm1/w(x) \). Thus, we will only describe solutions that correspond to \( u(x) \), real solutions of \eqref{sinh-gordon PIII}, via
\begin{equation}
\label{u to w}
w(x) = e^{u(x)/2}.
\end{equation}
As explained in \eqref{p parametrization}, real solutions of \eqref{sinh-gordon PIII} are parametrized by pairs \( (\iota,p)\), where \( \iota\in\{\pm1\}\) and \( p\in\overline\C\setminus\overline{\mathbb D}\). Moreover, as pointed out in the remark at the end of Section~\ref{2.2}, \( u(x;-\iota,p) = u(x;\iota,p) + (2n+1)2\pi i\), \( n\in\Z \). This means that
\begin{equation}
\label{iota to -iota}
w(x;-1,p) = -w(x;1,p)
\end{equation}
and therefore it is enough to consider the case \( \iota = 1 \) as was done in the proof of Theorem~\ref{main thm}. Next, one can readily see using \eqref{Stokes matrices} and \eqref{E parametrization} that \( \sigma_3 \Hat \Psi(\lambda) \sigma_3\) solves RHP~\ref{rhp original} when \( \Hat \Psi(\lambda) \) does, but with the monodromy data \( (a,A,B^\R) \) replaced by \( (-a,A,-B^\R)\), that is, \( (\iota,p)\) replaced by \( (-\iota,-p) \) when \( p \) is finite and \( (\iota,s^\R) \) replaced by \( (\iota,-s^\R) \) when \( p=\infty \). One can also readily check that \( \sigma_3P_0(u)\sigma_3 = P_0(-u)\), see \eqref{defn of P0}. Hence,
\begin{equation}
\label{p to -p}
w(x;1,-p) = \frac{-1}{w(x;1,p)} \quad \text{and} \quad w(x;1,\infty,-s^\R) = \frac{1}{w(x;1,\infty,s^\R)}
\end{equation}
when \( p \) is finite and infinite, respectively (the case \( p=\infty \) and \( s^\R=0\) corresponds to the trivial solutions \( w(x) = \iota \)). In particular, it is sufficient for us to consider only solutions \( w(x) \) with \( s^\R\leq 0 \) (\( \sigma=1\) when \(s^\R<0\)). Altogether, below we assume  \eqref{u to w} and take \( \iota=1 \), \( s^\R\leq 0 \).

We get from \eqref{gl Y asymptotic eq 2} and \eqref{gl Y asymptotic eq 3} as well as \eqref{eu2} that
\begin{equation}
\label{w at infinity}
w(x) \sim \begin{cases} \displaystyle 1 + (s^\R/\pi) K_0(x), & p=\infty, \medskip \\
\displaystyle \frac{1-\sin(x-\kappa \ln x + \phi)}{\cos(x-\kappa \ln x+\phi)}, & p\neq\infty, \end{cases}
\end{equation}
as \( x\to + \infty\). The corresponding formulae as \( x\to 0^+ \) now need to be deduced from \cite[Section~3]{Niles}. As this reference does not include final computations, we felt compelled to provide them below. To this end, we recall that \( u(x) = \pi i - i v(x) \), where \( v(x) \) is the solution of \eqref{sine-gordon PIII}. Then
\[
w(x) = e^{u(x)/2} = i e^{-iv(x)/2}.
\]
Let \( P_{\sin} \) be as in \eqref{Phi hat as lambda -> 0}. When \(s^\R \neq -2 \), it is stated on \cite[pages 60 and 63]{Niles} that
\[
P_{\sin} \sim \frac\pi8\frac1{\cos^2(\pi\alpha)}\begin{pmatrix} -i & 1 \\ i & 1\end{pmatrix} \begin{pmatrix} 0 & u_{-\alpha} \\ u_\alpha & 0 \end{pmatrix} \left(\frac x4\right)^{2\alpha\sigma_3} \begin{pmatrix} 0 &  m \\ \tilde m & 0 \end{pmatrix} d^{-\sigma_3} \begin{pmatrix} 0 & u_{-\alpha} \\ u_\alpha & 0 \end{pmatrix}\begin{pmatrix} 1 & 1 \\ -i & i \end{pmatrix},
\]
where we will introduce the symbols \( \alpha,\tilde m \), and \( d\) further below, while
\[
u_\nu = \frac{2^{\frac12-\nu}}{\Gamma(\frac12+\nu)}, \quad m = i \begin{cases} -2\cos(\pi\alpha), & p=\infty, \\ \displaystyle \frac{2-\overline p e^{i\pi\alpha} - p e^{-i\pi\alpha}}{\sqrt{|p|^2-1}}, & p\neq\infty, \end{cases}
\]
see \cite[pages 56-57]{Niles} and recall that we take \( 1/\sqrt{1+p_{\sin}q_{\sin}} =  i/\sqrt{|p|^2-1} \). Multiplying out the above product gives
\[
P_{\sin} \sim \frac\pi 8 \frac1{\cos^2(\pi\alpha)} \begin{pmatrix}
    mu_\alpha^2d(x/4)^{2\alpha} - \tilde m u_{-\alpha}^2d^{-1}(x/4)^{-2\alpha} & * \smallskip \\
    mu_\alpha^2d(x/4)^{2\alpha} + \tilde m u_{-\alpha}^2d^{-1} (x/4)^{-2\alpha} & *
\end{pmatrix},
\]
which, in turn, yields that
\begin{equation}
\label{solution w}
w(x) \sim \frac{\pi i}{4} \frac{mu_\alpha^2d}{\cos^2(\pi\alpha)} \left( \frac x4\right)^{2\alpha}.
\end{equation}

Let now \( s^\R\in(-2,0] \). In this case \( d=1 \) and \( \alpha=\gamma/4 \), where \( \gamma \) was introduced in Theorem~\ref{Niles result}, see \cite[pages 55 and 60]{Niles}. Recall further that we set \( q_p = \Re(p)+ \sqrt{1-\Im(p)^2}\) in this case (when \( p\neq \infty\)). Hence, by noticing that \( \cos(\pi\alpha) = \sqrt{1-\Im(p)^2} \) and after some straightforward simplifications, we get that
\[
m = -2i\cos(\pi\alpha) \frac{\sqrt{|p|^2-1}}{q_p},
\]
where we treat the last fraction as \( 1 \) when \( p=\infty \). Thus, using Euler's reflection formula for the gamma function, we can conclude that
\begin{equation}
\label{solution w -2<s<=0}
w(x) \sim \left( \frac x8 \right)^{\frac\gamma2} \frac{\Gamma(\frac12-\frac\gamma4)}{\Gamma(\frac12+\frac\gamma4)}\frac{\sqrt{|p|^2-1}}{q_p}.
\end{equation}

Let now \( s^\R<-2 \). In this case \( \alpha = \frac12 + it \), where \( t \) was introduced in Theorem~\ref{Niles result}, and
\[
d = 1 - \left(\frac 8x\right)^{4it}\frac{\Gamma^2(1+it)}{\Gamma^2(1-it)}\frac{\tilde m}{m}, \quad \tilde m = -i \begin{cases} 2\cos(\pi\alpha), & p=\infty, \\ \displaystyle \frac{2+\overline p e^{-i\pi\alpha} + p e^{i\pi\alpha}}{\sqrt{|p|^2-1}}, & p\neq\infty, \end{cases}
\]
see \cite[pages 61 and 63]{Niles}. Thus, it holds that
\[
m = -2i \cos(\pi\alpha) \frac{q_p}{|q_p|} \quad \text{and} \quad \tilde m = -2i \cos(\pi\alpha) \frac{\overline q_p}{|q_p|},
\]
where \( q_p = \Re(p) + i \sqrt{\Im(p)^2-1} \) (and therefore \( |q_p|=\sqrt{|p|^2-1} \)) and we treat the above fractions as \( 1 \) when \( p=\infty \). Hence, it holds that
\[
d = \left( \frac 8x \right)^{2it} \frac{\Gamma(1+it)}{\Gamma(1-it)} \frac{\overline q_p}{|q_p|} \, 2i \sin \left( 2t \ln \frac x8 + \Delta_p\right),
\]
where \( \Delta_p \) was introduced in Theorem~\ref{Niles result}. Consequently, we get from \eqref{solution w} that
\begin{equation}
\label{solution w s<-2}
w(x) \sim -\frac{x}{4t}\sin \left( 2t \ln \frac x8 + \Delta_p\right),
\end{equation}
where we used the identity \( |\Gamma(1+it)|^2 = \pi t/\sinh(\pi t)\).

Finally, let \( s^\R = -2 \), i.e., \( p=\Re(p)+i \). Then, it holds that
\[
P_{\sin} \sim \frac\pi{32}\begin{pmatrix} 1 & \frac 2\pi \\ 1 & -\frac 2\pi\end{pmatrix} \begin{pmatrix}
    m(x) & 0 \\ 0 & \tilde m(x) \end{pmatrix} \begin{pmatrix} 1 & 1 \\ \frac2\pi & -\frac2\pi \end{pmatrix} = \frac \pi{32} \begin{pmatrix} m(x) + \frac4{\pi^2} \tilde m(x) & * \\  m(x) - \frac4{\pi^2} \tilde m(x) & * \end{pmatrix},
\]
see \cite[pages 64-67]{Niles}, where
\[
m(x) = \frac{8i x}{\pi}\left(\ln\frac x8 + \delta_\infty \right) + \frac{4i x}{\overline p+i}
\]
(as in Theorem~\ref{Niles result}, \( \delta_\infty \) is the Euler's constant) and the explicit expression for \( \tilde m(x) \) is not important to us. Therefore, exactly as in the previous cases it holds that
\begin{equation}
\label{solution w s=-2}
w(x) \sim \frac{\pi i}{16}m(x) = -\frac x2\left(\ln\frac x8+\delta_p\right), \quad \delta_p = \delta_\infty + \frac{\pi}{2\Re(p)}.
\end{equation}
Formulae \eqref{w at infinity}, \eqref{solution w -2<s<=0}--\eqref{solution w s=-2} together with symmetries \eqref{iota to -iota} and \eqref{p to -p} finish the description of the connection formulae for the real solutions (on \( \R_{>0}\)) of Painlev\'e III($D_6$) equation \eqref{PIIID6} with $(\alpha, \beta, \gamma, \delta) = (0,0,1/4, -1/4)$.

\section{On Li's Conjecture} \label{conj by Li}

The conjecture raised in \cite{YuqiLi} concerns the connection between behavior at infinity and near zero of the radial solutions of \eqref{tt*-toda A3}. In the special case \( w_0=-w_1 \) it is assumed that a solution of \eqref{tt*-toda A3} satisfies
\begin{equation}
\label{Li assumption}
w_0(r) \sim s_2^\R \, 2^{-5/2} \, (\pi r)^{-1/2} e^{-4r},
\end{equation}
as \( r\to+\infty\), from which the behavior at the origin is inferred. More precisely, set
\[
\cos\left(\frac\pi2 (1+\gamma_0)\right) := \frac{s_2^{\R}}{2} \quad \text{and} \quad e^{-\rho_0} := 2^{ \gamma_0} \frac{\Gamma \left(\frac{1 + \gamma_0}{2} \right)}{\Gamma \left(\frac{1 - \gamma_0}{2} \right)},
\]
see \cite[Equations (2.2) and (4.1)]{YuqiLi} (\(\gamma_1=-\gamma_0\) and \( s_1^\R = 0 \) in the considered special case, where we also used Legendre duplication formula). It is then claimed, see cases \(\Omega_2,\Omega_5 \) of \cite[Conjecture~4.1]{YuqiLi}, that
\begin{equation}
\label{Li conclusion}
|s_2^\R|>2 \quad \Rightarrow \quad e^{2\sigma_2 w_0(r)} \sim 2\Re\big( e^{\sigma_2(\gamma_0\ln r+\rho_0)} \big)
\end{equation}
as \(r\to 0^+ \) with \(  e^{2w_0(r)} \) being smooth on \( \R_{>0} \), where \( \sigma_2 := -\sign(s_2^\R) \). 

Clearly, in the considered special case, \( w_0(r) \) is a solution of the tt\( ^*\)-Toda equation \eqref{tt*-Toda n=1} and therefore \( 4w_0(x/4) \) solves the sinh-Gordon Painlev\'e III equation \eqref{sinh-gordon PIII}. Assumption \eqref{Li assumption} and Theorem~\ref{main thm} immediately imply that \( 4w_0(x/4) \) is the solution corresponding to \( p=\infty \) and \( s^\R = s_2^\R \). It then follows that
\[
\gamma_0 = 2i t + \sigma, \quad e^{\sigma\gamma_0 \ln r} = re^{\sigma i 2t \ln r}, \quad \text{and} \quad e^{\sigma\rho_0} = \frac{\sigma i}{2t} e^{\sigma i(-2t\ln 2 + \Delta_\infty)},
\]
where \( t,\Delta_\infty,\sigma=\sigma_2 \) are the same as in Theorem~\ref{Niles result} and we used conjugate symmetry of the Gamma function. Hence,
\[
e^{2\sigma w_0(r)} \sim -\sigma \frac rt \Im\left( e^{\sigma i(2t \ln \frac r2 + \Delta_\infty)} \right)= - \frac rt \sin\left(2t \ln \frac r2 + \Delta_\infty\right),
\]
according to \eqref{Li conclusion}. Since \( 2w_0(r) = \frac12 u(4r) \), the conjecture now follows from   \eqref{u to w}, the second formula in \eqref{p to -p}, and \eqref{solution w s<-2}.

Of course, Theorems~\ref{main thm} and~\ref{Niles result} provide the connection between the behavior at infinity and zero of the radial solutions of \eqref{tt*-toda A3} when \( w_0 = - w_1 \) and \(|s_2^\R|\leq 2 \). However, these results are not new, see \cite{GIL1, GIL2, GIL3}, and were not part of \cite[Conjecture~4.1]{YuqiLi}.

%%%%%%%%%%%%%%%%%%%%%%%%%%%%%%%%%%%%%%%%%%%%%%%%%%%%%%%%%%%%%%%%%%%%%%%%%%%%%

\bibliographystyle{alpha}
\bibliography{reference.bib}

\end{document}